\documentclass[12pt]{article}

\usepackage[margin=1in]{geometry}
\usepackage{amsmath,amssymb,amsthm}
\usepackage{booktabs}
\usepackage[authoryear,round]{natbib}
\usepackage[hidelinks,colorlinks=true,linkcolor=blue,citecolor=blue,
            urlcolor=blue]{hyperref}
\usepackage{setspace}
\usepackage{graphicx}
\usepackage{xcolor}
\usepackage{listings}
\usepackage{caption}
\usepackage{enumitem}
\usepackage{float}
\usepackage{array}
\usepackage{multirow}

\theoremstyle{plain}
\newtheorem{theorem}{Theorem}[section]
\newtheorem{lemma}[theorem]{Lemma}
\newtheorem{proposition}[theorem]{Proposition}

\theoremstyle{definition}
\newtheorem{assumption}{Assumption}
\newtheorem{definition}{Definition}[section]
\newtheorem{remark}{Remark}

\newcommand{\E}{\mathbb{E}}
\newcommand{\Var}{\operatorname{Var}}
\newcommand{\Cov}{\operatorname{Cov}}
\newcommand{\difd}{\Delta}
\newcommand{\difp}[1]{\Delta^{#1}}
\newcommand{\ATTgt}{\mathrm{ATT}(g,t)}

\newcommand{\pgvec}{\mathbf{p}}
\newcommand{\iid}{\overset{\mathrm{iid}}{\sim}}
\newcommand{\xrightarrowd}{\xrightarrow{d}}

\lstdefinestyle{stata}{
  language={}, basicstyle=\footnotesize\ttfamily,
  backgroundcolor=\color{gray!8}, frame=single, framesep=4pt,
  breaklines=true, keepspaces=true,
  commentstyle=\color{gray!60},
  morecomment=[l]{*}, morecomment=[l]{//},
  keywordstyle=\color{blue!80}\bfseries,
  morekeywords={ssc,use,gen,replace,keep,drop,tab,sum,reg,csdid,
    estat,event_plot,xtset,quietly,qui,di,list,clear,set,egen,
    bsample,forvalues,foreach,local,global,save,collapse,bys,
    duplicates,sort,predict,append,merge,describe},
  numbers=left,numberstyle=\tiny\color{gray},
  stepnumber=1,numbersep=6pt,showstringspaces=false,tabsize=4
}

\begin{document}

\begin{center}
  {\LARGE\bfseries Beyond Parallel Trends in Staggered 
   Difference-in-Differences: 
   Identification under Higher-Order Parallelism\par}

  \vspace{0.6cm}
  {\large Zecharias Anteneh, Ph.D$^{\dagger}$\par}
  \vspace{0.2em}
  {\normalsize Centre for Health Economics, University of York\par}
  \vspace{0.2em}
  {\normalsize \texttt{zecharias.anteneh@york.ac.uk}\par}
  \vspace{0.2em}
  {\normalsize\itshape June 2026\par}
\end{center}

\footnotetext[1]{I am extremely grateful to Jonathan Roth for detailed comments that improved the paper. I thank Scott Cunningham for preparing and making publicly available the cleaned Medicaid expansion panel used in the empirical application; the underlying microdata are from the American Community Survey. All errors are
my own.}

\vspace{0.4cm}
  
  \begin{abstract}\noindent  
  In difference-in-differences designs, the parallel trends assumption requires that the outcome gap between treated and control units would have remained flat absent treatment. Pre-treatment event studies frequently reject this flat-gap requirement. Existing responses include parametric trend controls and bounds on the treatment effect under assumptions about the magnitude of the violation. This paper shows that point identification of cohort-specific and aggregate treatment effects in staggered designs remains achievable under strictly weaker assumptions. I replace the flat-gap requirement with a hierarchy of higher-order conditions, Parallel[$p$], embed this framework in the group-time average treatment effect structure of \citet{callaway2021difference}, and prove an aggregation theorem for the case where different cohorts are identified under different feasible polynomial orders, a challenge unique to staggered designs that has not been previously addressed (Theorem~\ref{thm:aggregation}). A sequential order-selection procedure guides applied practice. Monte Carlo evidence confirms that post-selection bootstrap coverage remains near-nominal and that inference is robust to realistic serial correlation. Applied to Medicaid expansion data, the method yields point estimates resting on an assumption the pre-treatment data do not reject, in contrast to the flat-gap requirement which those same data decisively reject.

\medskip
\noindent\textbf{Keywords:} difference-in-differences, staggered
adoption, parallel trends, higher-order parallelism, group-time
ATT, Medicaid expansion.

\noindent\textbf{JEL codes:} C21, C23.
\end{abstract}

\section{Introduction}
\label{sec:intro}

\indent Difference-in-differences (DiD) is among the most widely used designs for program and policy evaluation. In staggered adoption settings, recent work has shown that the conventional two-way fixed effects regression can deliver misleading estimates when treatment effects vary across cohorts, and has proposed alternative estimators that recover interpretable average treatment effects \citep{callaway2021difference,sun2021estimating,goodman2021difference,de2020two,borusyak2024revisiting}. All of these frameworks, however, rest on the same core identifying assumption: \emph{parallel trends}.

Parallel trends requires that, absent treatment, the average outcome gap between treated and control units would have remained constant over time. In practice, pre-treatment event studies frequently show systematic trends that are difficult to reconcile with this requirement, especially when pre-periods are long. When this happens, treatment effect estimates are biased, and the magnitude of that bias is generally unknown. The common response is to test for pre-trends and proceed only if the test passes, but this practice is itself problematic: pre-trends tests are often underpowered, and conditioning the analysis on having passed one distorts subsequent inference \citep{roth2022pretest}. A researcher who finds a trending pre-period is left without a clear way forward.

This paper asks: when parallel trends fails in a staggered adoption design, can we still point-identify average treatment effects without discarding some cohorts or bounding the effect under assumptions about the magnitude of the violation? I propose a middle path. Standard parallel trends implicitly assumes that whatever shape the pre-treatment gap takes, it would have become flat after treatment. I instead use the information in the pre-treatment trajectory: when the gap evolves smoothly, that evolution is informative about what it would have done absent treatment. I replace Parallel[1] with a hierarchy of higher-order assumptions, Parallel[$p$], and project the pre-treatment trajectory forward as the counterfactual. Where \citet{rambachan2023more} trace out an identified set indexed by the strength of a smoothness restriction, I instead select an order from the pre-treatment data and report a point estimate under that order.

Formally, Parallel[$p$] requires only that the $p$-th order time difference of the untreated outcome gap is common across groups. Parallel[1] is the standard assumption: the gap must be flat. Parallel[2] allows the gap to trend linearly; Parallel[3] allows a quadratic trajectory; higher orders allow richer curvature. Each step up in $p$ is strictly weaker than the previous one and remains sufficient for point identification of group-time average treatment effects. When a cohort has more than $p$ pre-treatment periods, Parallel[$p$] generates testable implications about the pre-treatment gap that can be assessed with observable data.

These implications do not, however, eliminate extrapolation. Every DiD estimator relies on an extrapolation assumption: Parallel[1] extrapolates a flat gap into the post-treatment period, while Parallel[$p$] extrapolates the pre-existing polynomial trajectory. Neither is testable in the post-treatment period. The advantage of Parallel[$p$] is that its key implication, that the pre-treatment gap lies on a low-degree polynomial, is testable with observed data, whereas the flatness required by Parallel[1] is often visibly violated. Parallel[$p$] thus replaces an assumption the data may reject with one they do not, recovering point identification where standard methods cannot proceed.

The paper makes three contributions. First, I formally identify $\ATTgt$ under Parallel[$p$] for arbitrary $p \geq 1$ in the staggered adoption framework of \citet{callaway2021difference}, using the polynomial structure of the pre-treatment gap to construct a cohort-specific counterfactual in each post-treatment period.

Second, I address a problem that arises only in staggered settings. Because treatment timing varies, cohorts differ in how many pre-treatment periods they have, and therefore in the maximum order each can support: a cohort treated two periods into the panel can support only Parallel[1], while one treated ten periods in can support up to Parallel[9]. The natural question is how to combine cohorts that are identified under different orders into a single summary effect. Theorem~\ref{thm:aggregation} resolves this, showing that coherent weighted averages of group-time ATTs remain identified even when each cohort $g$ is identified under its own order $p_g$. To my knowledge, this is the first aggregation result for staggered DiD under cohort-heterogeneous orders.

Third, I develop a sequential order-selection procedure tailored to staggered settings, provide Monte Carlo evidence on finite-sample performance and post-selection inference, and implement the method in a companion Stata command \texttt{anddp},
which produces the aggregate DD[$p$] ATT with cluster bootstrap inference. The selection procedure responds directly to the pre-testing critique: rather than testing the flat-gap null and conditioning on passing it, it selects the lowest order the pre-treatment data do not reject and estimates under that order. The Monte Carlo evidence shows that bootstrap confidence intervals retain near-nominal coverage after this selection step.

I apply the method to the staggered adoption of Medicaid expansion under the Affordable Care Act, using state-level insurance coverage data. Expanding and non-expanding states were on visibly diverging insurance trajectories before any expansion took effect, and a joint pre-trends test rejects flatness decisively. Standard \textit{DiD}, which assumes the gap would have stayed flat, is therefore hard to defend here. The proposed estimator, DD[$p$], nonetheless recovers point estimates under Parallel[$p$], an assumption the same pre-treatment data do not reject, and yields: expansion increased insurance coverage by roughly six percentage points.

The higher-order parallelism idea builds on \citet{mora2019alternative}, who develop the Parallel[$p$] hierarchy and prove identification in a two-group, multiple-period setting. This paper extends their framework to staggered adoption. A closely related paper is \citet{egami2023using}, who propose a GMM estimator combining first- and second-difference moment conditions for staggered designs, and whose online appendix sketches a $k$-th order generalization structurally similar to Parallel[$p$]. That generalization is not applied in their own empirical analysis, is not exposed in their companion software, and -- unlike the present paper -- does not develop a discrete order-selection procedure or a formal aggregation result for cohorts identified under different orders; their estimator instead forms a continuous GMM-efficient blend across orders rather than selecting and reporting a single order per cohort. The present paper develops the full Parallel[$p$] hierarchy, proves identification of ATT(g,t) under each order, and establishes an aggregation theorem (Theorem~\ref{thm:aggregation}) for cohorts identified under heterogeneous orders, together with a sequential selection procedure and supporting software.

Prior responses to parallel trends violations include the following. Sensitivity bounds \citep{rambachan2023more} restrict how far the post-treatment violation of parallel trends can depart from the pre-treatment trend, and report how the identified set for the treatment effect widens as that restriction is relaxed; the approach is rigorous and transparent, and rather than committing to a single counterfactual it traces out a range of estimates indexed by the strength of the assumption. The closest connection is to their smoothness restriction: assuming the post-treatment violation is exactly linear coincides with Parallel[2] in the present hierarchy. The approaches diverge in what they do with that restriction. They treat linearity as one end of a continuum and report how the identified set grows as it is relaxed, whereas I select an order from the data and report a point estimate under it. Pre-trend tests \citep{roth2022pretest} diagnose violations but do not provide a corrected estimator. A complementary approach is the non-inferiority framework of \citet{bilinski2026nothing}, which recasts pre-trend assessment as an equivalence test: rather than asking whether pre-trends are exactly zero, it asks whether deviations are small enough not to matter for the estimated effect. Their procedure operates under Parallel[1] and provides a decision rule for tolerating small violations. \citet{roth2023s} provide a comprehensive review of these and other responses to parallel trends violations.

A related practice in applied work is to include unit-specific linear time trends in TWFE regressions to absorb differential pre-trends \citep{angrist2009mostly}. \citet{dobkin2018economic} adopt a conceptually related strategy in a single-timing event study: they include a linear trend in event time alongside a saturated set of post-treatment dummies. Because the post-treatment dummies absorb all post-treatment variation, the slope coefficient is identified from pre-treatment data only. This is formally equivalent to the $p=2$ case of the present framework applied to a single cohort, where no aggregation is required. Dobkin's linear extrapolation is therefore a single one-step relaxation of flat parallel trends. DD[$p$], the estimator I develop under Parallel[$p$], generalises this in three ways: it extends to arbitrary order $p \geq 1$ with a sequential procedure to select the order the data support; it handles staggered adoption with multiple cohorts; and it provides formal aggregation across cohorts identified under different feasible orders, together with asymptotic inference (see Appendix~\ref{app:comparison}).

Other approaches to parallel-trends failures include synthetic control methods \citep{abadie2010synthetic,xu2017generalized,arkhangelsky2021synthetic,ben2022synthetic}, which construct a weighted comparison group that matches treated units' pre-treatment trajectories and work well with long pre-treatment series and relatively few treated units. Triple-differences designs \citep{strezhnev2023decomposing,ortiz2025better} require a placebo stratum known to be unaffected by treatment. Change-in-changes \citep{athey2006identification} relies on rank preservation. Partial identification approaches \citep{manski2018right} deliver bounds on treatment effects. Each of these methods is suited to a different empirical context; the present approach targets the common setting where the pre-treatment gap evolves smoothly and a polynomial counterfactual is credible.

The remainder of the paper is organised as follows. Section~\ref{sec:setup} establishes the framework. Section~\ref{sec:hotpt} introduces the hierarchy of higher-order parallelism assumptions. Section~\ref{sec:identification} derives the identification results, including the aggregation theorem. Section~\ref{sec:est} presents estimation and inference. Section~\ref{sec:orderselection} develops the order-selection procedure. Section~\ref{sec:simulation} reports Monte Carlo evidence. Section~\ref{sec:empirical} presents the Medicaid expansion application. Section~\ref{sec:conc} concludes.

\section{Setup}
\label{sec:setup}

\subsection{Panel, Treatment, and Notation}

The notation and staggered-adoption setup follow \citet{callaway2021difference}. Consider a balanced panel of $N$ units over $T$ periods $t = 1,\ldots,T$. Unit $i$ has adoption date $G_i \in \{g_1,\ldots,g_K,\infty\}$. $G_i = g$ means unit $i$ is first treated in period $g$; $G_i = \infty$ means never treated. Treatment is absorbing: once a unit is treated it remains treated in all subsequent periods, so a unit's adoption date $G_i$ fully summarises its treatment path. Units with $G_i = g$ form \emph{cohort} $g$. Let $N_g$ be the size of cohort $g$ and $N_\infty$ the number of never-treated units, with $N = \sum_g N_g + N_\infty$.

Asymptotics are large-$N$ with $T$ fixed. The cohort fractions $\pi_g = N_g/N$ are held fixed as $N$ grows, which is achieved by requiring each $N_g$ to grow proportionally with $N$, a standard large-sample device that prevents any cohort from vanishing or dominating in the limit.

The number of pre-treatment periods for cohort $g$ is:
\begin{equation}
  m_g = g - t_{\min}, \label{eq:mg}
\end{equation}
where $t_{\min}$ is the first observed period.

\subsection{Potential Outcomes}

For each unit $i$ and period $t$, let $Y_{it}(g)$ denote the potential outcome if first treated in period $g$, and
$Y_{it}(\infty)$ the never-treated potential outcome \citep{neyman1923application,rubin1974estimating}.
The observed outcome is:
\begin{equation}
  Y_{it} = \begin{cases}
    Y_{it}(\infty) & t < G_i \\
    Y_{it}(G_i)   & t \geq G_i.
  \end{cases}
\end{equation}

Before treatment, observed outcomes equal untreated potential outcomes, making pre-treatment data informative about the
untreated trajectory.

\begin{assumption}[No Anticipation]\label{ass:na}
For all $i$, $g$, and $t < g$: $Y_{it}(g) = Y_{it}(\infty)$.
\end{assumption}

Units do not alter behaviour in anticipation of future treatment, so their pre-treatment observed outcomes equal their untreated
potential outcomes.

\begin{assumption}[Overlap]\label{ass:ov}
For each cohort $g$, the probability of belonging to that cohort is strictly positive and strictly less than one: 
$0 < \Pr(G_i = g) < 1$. \end{assumption}

\subsection{Target Parameters}

The group-time average treatment effect is:
\begin{equation}
  \ATTgt = \E[Y_{it}(g) - Y_{it}(\infty) \mid G_i=g],
  \quad t \geq g. \label{eq:attgt}
\end{equation}
The expectation averages over all units $i$ belonging to cohort $g$, comparing their actual post-treatment outcome to what they
would have experienced had treatment never occurred.

A scalar summary is the weighted aggregate:
\begin{equation}
  \theta = \sum_{g}\sum_{t \geq g} w_{g,t}\,\ATTgt, \label{eq:theta}
\end{equation}
for non-negative weights $w_{g,t}$ summing to one.

\section{Higher-Order Parallel Trends}
\label{sec:hotpt}

\subsection{Standard Parallel Trends}

\begin{assumption}[Parallel{[1]}]\label{ass:pt1}
For all cohorts $g$ and periods $t$:
$\E[\difd Y_{it}(\infty)\mid G_i=g]
= \E[\difd Y_{it}(\infty)\mid G_i=\infty]$,
where $\difd Y_{it} = Y_{it} - Y_{i,t-1}$.
\end{assumption}

Standard parallel trends says that, absent treatment, the average untreated outcome would have changed by the same amount each period for every group, treated cohorts and never-treated alike. Equivalently, the gap between any cohort and the never-treated group would have stayed constant over time. This has a testable pre-treatment implication: the pre-treatment gap should be flat. When the event-study plot shows a visible pre-treatment trend, that implication fails, and Parallel[$1$] is implausible.

\subsection{Higher-Order Differences}

The $p$-th order difference operator is defined recursively:
$\difp{1}Y_{it} = Y_{it} - Y_{i,t-1}$ and
$\difp{p}Y_{it} = \difp{1}(\difp{p-1}Y_{it})$ for $p \geq 2$.
Written out:
\begin{equation}
  \difp{p}Y_{it}
  = \sum_{k=0}^{p}(-1)^k\binom{p}{k}Y_{i,t-k}. \label{eq:pdiff}
\end{equation}
$\difp{1}$ measures the year-on-year change; $\difp{2}$ measures how that change is itself changing; $\difp{3}$ measures the rate
of change of $\difp{2}$, and so on.

To build intuition, suppose the gap between treated and control groups rises by 0.01 per year before treatment. Parallel[1] requires the gap to be constant, so a steadily rising gap is inconsistent with it. Parallel[2] requires only that this 0.01 annual change is common to the treated and never-treated groups, a weaker restriction that the pre-treatment data may well support. 

\begin{assumption}[Parallel{[$p$]}]\label{ass:ptp}
For a given integer $p \geq 1$, all cohorts $g$, and all periods $t$:
\[
  \E[\difp{p}Y_{it}(\infty)\mid G_i=g]
  = \E[\difp{p}Y_{it}(\infty)\mid G_i=\infty].
\]
\end{assumption}

Parallel[1] is the special case $p=1$. Parallel[2] allows different levels and slopes across groups, requiring only that the acceleration is common. Each step up in $p$ is a strictly weaker assumption and requires one additional pre-treatment period to test.

\begin{remark}[Parallel{[$p$]} does not imply Parallel{[$p-1$]}]
  The hierarchy is nested in the sense that data satisfying
  Parallel[$p-1$] also satisfy Parallel[$p$] (a flat function
  is a special case of a linear function).
  But the converse does not hold: data with a stable linear
  pre-trend satisfy Parallel[2] but not Parallel[1].
  The appropriate order is determined by the data-generating
  process and estimated via the sequential test in
  Section~\ref{sec:orderselection}.
\end{remark}

\section{Identification in Staggered Settings}
\label{sec:identification}
\subsection{Pre-Treatment Gap and Polynomial Structure}

Define the observed gap between cohort $g$ and never-treated
units in period $t$:
\begin{equation}
  \gamma_{g,t}
  = \E[Y_{it}\mid G_i=g] - \E[Y_{it}\mid G_i=\infty].
  \label{eq:gap}
\end{equation}
Under Assumption~\ref{ass:na}, this equals the untreated
potential outcome gap in the pre-treatment period.

\begin{lemma}[Polynomial Gap Structure]
  \label{lem:poly}
  Under Assumptions~\ref{ass:na} and \ref{ass:ptp},
  for all pre-treatment periods $t < g$ with $t \geq t_{\min}+p$:
  $\difp{p}\gamma_{g,t} = 0$.
  Equivalently, $\gamma_{g,t}$ is a polynomial of degree $p-1$
  in $t$ for all $t < g$:
  \begin{equation}
    \gamma_{g,t}
    = c_{g,0} + c_{g,1}t + \cdots + c_{g,p-1}t^{p-1},
    \quad t < g. \label{eq:polygap}
  \end{equation}
  The cohort-specific constants $c_{g,0},\ldots,c_{g,p-1}$ are
  identified from $p$ pre-treatment gap observations.
\end{lemma}

\begin{proof}
Under Assumption~\ref{ass:na}, for $t < g$:
$\E[Y_{it}\mid G_i=g] = \E[Y_{it}(\infty)\mid G_i=g]$.
Therefore:
\[
  \gamma_{g,t}
  = \E[Y_{it}(\infty)\mid G_i=g]
  - \E[Y_{it}(\infty)\mid G_i=\infty].
\]
Since $\difp{p}$ is a linear operator:
\begin{align*}
  \difp{p}\gamma_{g,t}
  &= \E[\difp{p}Y_{it}(\infty)\mid G_i=g]
   - \E[\difp{p}Y_{it}(\infty)\mid G_i=\infty] = 0,
\end{align*}
where the equality uses Assumption~\ref{ass:ptp}.
A sequence on $\mathbb{Z}$ with identically zero $p$-th difference is a polynomial of degree at most $p-1$, the discrete analogue of the fact that a function whose $p$-th derivative is identically zero is a polynomial of degree $p-1$.\footnote{This is a standard result in finite-difference calculus; see, e.g., Jordan (1965, \textit{Calculus of Finite Differences}) for a classical treatment.}
Formally, let $\mathcal{P}_{p-1}$ denote the space of polynomials of degree $\leq p-1$ on $\mathbb{Z}$.
The kernel of $\difp{p}$ on $\mathbb{Z}$-sequences is exactly
$\mathcal{P}_{p-1}$ (this follows by induction: $\difp{1}\gamma=0$ iff $\gamma$ is constant; if $\difp{p}\gamma=0$ then $\difp{p-1}(\difd\gamma)=0$, so $\difd\gamma\in\mathcal{P}_{p-2}$,
which implies $\gamma\in\mathcal{P}_{p-1}$).
The $p$ coefficients $c_{g,0},\ldots,c_{g,p-1}$ are uniquely
determined by any $p$ distinct values of $\gamma_{g,t}$ from
the pre-treatment period, obtained by solving the resulting linear system or, when $m_g > p$, by OLS.
\end{proof}

Lemma~\ref{lem:poly} makes the identifying content of Parallel[$p$] concrete and testable: the pre-treatment gap must lie on a polynomial of degree $p-1$. For $p=1$: flat (constant). For $p=2$: straight line.
For $p=3$: quadratic curve. The residuals from fitting this polynomial to the pre-treatment data are the testable implications used in
Section~\ref{sec:orderselection}.

\subsection{Counterfactual Identification}

The counterfactual gap, what the gap would have been in the post-treatment period absent treatment, is:
\begin{equation}
  \gamma_{g,t}(0)
  = \E[Y_{it}(\infty)\mid G_i=g]
  - \E[Y_{it}(\infty)\mid G_i=\infty],
  \quad t \geq g. \label{eq:cfgap}
\end{equation}
This is unobservable for $t \geq g$ because treated units observed outcomes include the treatment effect.

\begin{proposition}[Counterfactual Identification]
  \label{prop:cf}
  Under Assumptions~\ref{ass:na} and \ref{ass:ptp}, with
  $m_g \geq p$:
  \begin{enumerate}[label=(\roman*)]
    \item $\gamma_{g,t}(0)$ is a polynomial of degree $p-1$ in $t$
      for \emph{all} periods $t$, including post-treatment.
    \item This polynomial is uniquely identified by the $p$ most
      recent pre-treatment observations
      $\{\gamma_{g,g-k}\}_{k=1}^{p}$, which are observable under
      Assumption~\ref{ass:na}
      (see Remark~\ref{rem:ols_vs_interp} for the finite-sample
      estimator, which uses all $m_g$ pre-treatment observations
      for efficiency).
    \item Evaluating the polynomial at any $t \geq g$ gives
      $\gamma_{g,t}(0)$.
  \end{enumerate}
  In particular:
  for $p=1$, $\gamma_{g,t}(0) = \gamma_{g,g-1}$ (flat at last
  pre-period value, the standard DiD counterfactual);
  for $p=2$,
  $\gamma_{g,t}(0) = (t-g+2)\gamma_{g,g-1}-(t-g+1)\gamma_{g,g-2}$
  (linear extrapolation of the pre-existing trend).
\end{proposition}

\begin{proof}
  Part (i): Assumption~\ref{ass:ptp} states $\difp{p}\gamma_{g,t}(0) = 0$ for all $t$, not just pre-treatment periods. By the same argument as Lemma~\ref{lem:poly}, $\gamma_{g,t}(0) \in \mathcal{P}_{p-1}$ for all $t$.

  Part (ii): A polynomial of degree $p-1$ is uniquely determined by
  $p$ values. Under Assumption~\ref{ass:na}, $\gamma_{g,g-k} = \E[Y_{i,g-k}(\infty)\mid G_i=g] - \E[Y_{i,g-k}(\infty)\mid G_i=\infty]$ for $k=1,\ldots,p$, which equals $\gamma_{g,g-k}(0)$ (since no treatment has occurred). These $p$ values uniquely pin down the polynomial.\footnote{The coefficient vector $\mathbf{c}_g$ is uniquely determined because the $p \times p$ Vandermonde matrix formed from $p$ distinct integer time indices has full column rank; distinct values of $t$ guarantee non-vanishing Vandermonde determinant.}

  Part (iii): Evaluate the identified polynomial at $t \geq g$.
  Since $\gamma_{g,t}(0)$ is a polynomial everywhere and we have
  identified it from pre-treatment data, its post-treatment values
  are also identified.

  For $p=1$: the unique degree-0 polynomial through one point is
  the constant $\gamma_{g,g-1}(0) = \gamma_{g,g-1}$.
  For $p=2$: the unique degree-1 polynomial through two points
  $\{(g-2, \gamma_{g,g-2}), (g-1, \gamma_{g,g-1})\}$ has slope
  $\gamma_{g,g-1}-\gamma_{g,g-2}$ and evaluates at $t$ to
  $\gamma_{g,g-2} + (t-g+2)(\gamma_{g,g-1}-\gamma_{g,g-2})
  = (t-g+2)\gamma_{g,g-1}-(t-g+1)\gamma_{g,g-2}$.
\end{proof}

\subsection{Group-Time ATT Identification}

\begin{theorem}[Identification under Parallel{[$p$]}]
  \label{thm:identification}
  Under Assumptions~\ref{ass:na}, \ref{ass:ov},
  and \ref{ass:ptp} (for the relevant order $p$),
  with $m_g \geq p$,
  for any post-treatment period $t \geq g$:
  \begin{equation}
    \ATTgt = \gamma_{g,t} - \gamma_{g,t}(0), \label{eq:attid}
  \end{equation}
  where $\gamma_{g,t}$ is the observed post-treatment gap and
  $\gamma_{g,t}(0)$ is the polynomial counterfactual from
  Proposition~\ref{prop:cf}.
  Both are identified from observable data under
  Assumption~\ref{ass:ov}.
\end{theorem}

\begin{proof}
  Write:
  \begin{align*}
    \ATTgt
    &= \E[Y_{it}\mid G_i=g] - \E[Y_{it}(\infty)\mid G_i=g].
  \end{align*}
  By definition of $\gamma_{g,t}(0)$:
  $\E[Y_{it}(\infty)\mid G_i=g]
  = \E[Y_{it}(\infty)\mid G_i=\infty] + \gamma_{g,t}(0)$.
  Never-treated units are never treated, so
  $\E[Y_{it}(\infty)\mid G_i=\infty] = \E[Y_{it}\mid G_i=\infty]$,
  which is directly observable.
  Therefore:
  \begin{align*}
    \ATTgt
    &= \E[Y_{it}\mid G_i=g]
     - \E[Y_{it}\mid G_i=\infty]
     - \gamma_{g,t}(0) \\
    &= \gamma_{g,t} - \gamma_{g,t}(0). \qedhere
  \end{align*}
\end{proof}

\noindent
The treatment effect is the vertical distance between the observed
post-treatment gap and the polynomial counterfactual: the hollow
dots above the dashed line in Figure~\ref{fig:gap}.

\subsection{The Short Pre-Period Problem and Cohort-Heterogeneous
  Orders}

In a two-group setting, the number of pre-treatment periods is
fixed for the single treated cohort, and the researcher simply
chooses the largest $p$ the panel supports.
In staggered settings, $m_g = g - t_{\min}$ varies across cohorts.
Early-treated cohorts have small $m_g$; late-treated cohorts have
large $m_g$.

\begin{definition}[Feasible Set and Cohort-Specific Orders]
  \label{def:feasible}
  For target order $p$, the feasible cohort set is
  $\mathcal{F}(p) = \{g : m_g \geq p\}$.
  The cohort-specific applied order is
  $p_g \in \{1,\ldots,m_g-1\}$,
  the order assumed for cohort $g$.
\end{definition}

Orders up to $m_g$ are identified in population, since identification
requires only $m_g \geq p$, but $p_g = m_g$ exactly interpolates the
pre-treatment gaps and leaves no residual degree of freedom for
diagnostics; I therefore restrict attention to $p_g \leq m_g - 1$
throughout.

When $p_g < p$ for some cohorts, applying a uniform $p$ either
requires excluding those cohorts (changing the estimand) or
over-relaxing to a lower order than the data of later-treated
cohorts can support.
The following theorem resolves this by allowing heterogeneous
orders across cohorts.

\begin{theorem}[Aggregation under Cohort-Heterogeneous Orders]
  \label{thm:aggregation}
  Suppose for each cohort $g$,
  Assumptions~\ref{ass:na}--\ref{ass:ov} hold, and
  Assumption~\ref{ass:ptp} holds for cohort-specific order $p_g$,
  where $1 \leq p_g \leq m_g-1$.
  Let $\pgvec = (p_g)_g$.
  Then:
  \begin{enumerate}[label=(\roman*)]
    \item \emph{(Cohort identification)}
      $\ATTgt$ is identified for each $g$ and $t \geq g$ by
      Theorem~\ref{thm:identification} applied with order $p_g$.
    \item \emph{(Aggregate identification)}
      $\theta(\pgvec) = \sum_{g}\sum_{t\geq g}
      w_{g,t}\,\mathrm{ATT}^{(p_g)}(g,t)$ is identified
      for any non-negative weights summing to one.
    \item \emph{(Callaway--Sant'Anna as special case)}
      When $p_g = 1$ for all $g$ and the weights $w_{g,t}$ are
      chosen to match the \citet{callaway2021difference}
      aggregation scheme, $\theta(\pgvec)$ identifies the same
      population parameter as the Callaway--Sant'Anna aggregate ATT
      under standard parallel trends.
  \end{enumerate}
\end{theorem}

\begin{proof}
  Part (i) follows from Theorem~\ref{thm:identification} applied
  independently to each cohort under its own order $p_g$.
  Each cohort's counterfactual is constructed from that cohort's
  own pre-treatment gaps; although all cohorts share the same
  never-treated group as the comparison group, the Parallel[$p_g$]
  restriction for cohort $g$ is a population condition on cohort
  $g$'s potential outcomes and can be imposed independently of
  restrictions placed on other cohorts.
  Hence identification for one cohort does not require or restrict
  the identifying assumption of another.

  Part (ii): since each $\mathrm{ATT}^{(p_g)}(g,t)$ is identified
  (part~i), and a finite weighted average of identified quantities
  is identified, $\theta(\pgvec)$ is identified.
  Formally, write
  $\theta(\pgvec) = \boldsymbol{w}^\top \boldsymbol{a}$,
  where $\boldsymbol{a}$ collects all identified
  $\mathrm{ATT}^{(p_g)}(g,t)$ and $\boldsymbol{w}$ are the
  corresponding non-negative weights summing to one.
  Since each component of $\boldsymbol{a}$ is identified, so is
  $\theta$.

  Part (iii): when $p_g=1$ for all $g$,
  Proposition~\ref{prop:cf} gives
  $\gamma_{g,t}(0) = \gamma_{g,g-1}$ for all $t \geq g$.
  This is the flat counterfactual of standard parallel trends.
  When the weights $w_{g,t}$ are furthermore chosen to match the
  \citet{callaway2021difference} aggregation scheme, $\theta(\pgvec)$
  identifies the same population parameter as the Callaway--Sant'Anna
  aggregate ATT under standard parallel trends.
\end{proof}

\begin{remark}[Economic interpretation of the aggregate]
  The structural parameter is $\ATTgt$, defined in
  Section~\ref{sec:setup} without reference to any identifying
  order.
  The superscript $(p_g)$ on $\widehat{\mathrm{ATT}}^{(p_g)}(g,t)$
  denotes the estimator constructed under order $p_g$, not a
  different population quantity: the target is always $\ATTgt$.
  When different cohorts are identified under different orders
  $p_g$, the aggregate $\theta(\pgvec)$ remains interpretable
  because each cohort's $\ATTgt$ is a well-defined structural
  object, the average difference between actual and
  counterfactual outcomes for that cohort, regardless of how
  it is identified.
  The identification strategy does not change the parameter;
  it only changes the assumption under which it is recovered.
  Transparency requires reporting which order was applied to each
  cohort; the three reporting strategies in Section~\ref{sec:est}
  formalise this.
\end{remark}

\begin{remark}[Relation to prior work]
  \citet{mora2019alternative} prove the identification result in
  Theorem~\ref{thm:identification} for the two-group case under
  Parallel[$p$].
  \citet{callaway2021difference} provide aggregation of group-time
  ATTs under standard parallel trends ($p_g = 1$ for all cohorts).
 \citet{egami2023using} propose, but do not operationalize or apply, a $k$-th order generalization for staggered designs in their online appendix.
  Theorem~\ref{thm:aggregation} extends these results by allowing
  cohort-specific orders $p_g \in \{1,\ldots,m_g-1\}$; to my
  knowledge it is the first aggregation result for staggered DiD
  under cohort-heterogeneous feasible orders.
\end{remark}

\section{Estimation and Inference}
\label{sec:est}

\subsection{The DD[\texorpdfstring{$p$}{p}] Estimator}

Let $\bar Y_{g,t} = N_g^{-1}\sum_{i:G_i=g}Y_{it}$ and $\bar Y_{\infty,t}$ be the cohort and never-treated sample means.

\medskip\noindent
\textbf{Step 1. Sample gaps.}
$\hat\gamma_{g,t} = \bar Y_{g,t} - \bar Y_{\infty,t}$.

\medskip\noindent
\textbf{Step 2. Polynomial fit.}
Fit a polynomial of degree $p_g - 1$ by OLS to the $m_g$
pre-treatment gap observations:
\begin{equation}
  (\hat c_{g,0},\ldots,\hat c_{g,p_g-1})
  = \arg\min_{c} \sum_{t<g}
  \Bigl(\hat\gamma_{g,t}-\sum_{k=0}^{p_g-1}c_k t^k\Bigr)^2.
  \label{eq:ols}
\end{equation}
When $m_g = p_g$, this is exact interpolation; when $m_g > p_g$,
the OLS residuals provide the over-identifying restrictions used
in Section~\ref{sec:orderselection}.\footnote{To improve numerical stability, particularly at orders $p \geq 3$, time should be centred at the mean pre-treatment year before constructing the polynomial basis. Results are invariant to this reparametrisation but condition numbers are substantially reduced.}

\begin{remark}[OLS estimator versus minimum-point counterfactual]
  \label{rem:ols_vs_interp}
  Proposition~\ref{prop:cf} identifies the counterfactual from
  any $p_g$ pre-treatment observations, most naturally the $p_g$
  most recent.
  Step~2 uses all $m_g$ pre-treatment observations in an OLS fit.
  Under exact Parallel[$p_g$], both approaches recover the same
  population polynomial, but they are different finite-sample
  estimators when $m_g > p_g$: the OLS fit exploits all available
  pre-treatment data and is generally more efficient.
  In particular, at $p_g = 1$ the OLS estimator averages all
  pre-treatment gaps rather than using only the last pre-period
  observation, so it does not reduce to a base-period comparison
  in finite samples.
  Applied researchers wishing strict numerical equivalence to
  Callaway--Sant'Anna at $p_g=1$ should use \texttt{csdid}
  directly; the present estimator delivers the same identification
  but uses a different finite-sample implementation.
\end{remark}

\medskip\noindent
\textbf{Step 3 --- Counterfactual and ATT.}
$\hat\gamma_{g,t}(0) = \sum_{k=0}^{p_g-1}\hat c_{g,k}\,t^k$
and:
\begin{equation}
  \widehat{\mathrm{ATT}}^{(p_g)}(g,t)
  = \hat\gamma_{g,t} - \hat\gamma_{g,t}(0).
  \label{eq:atthat}
\end{equation}

\medskip\noindent
\textbf{Step 4 --- Aggregate.}
$\hat\theta(\pgvec) = \sum_{g,t\geq g}
\hat w_{g,t}\,\widehat{\mathrm{ATT}}^{(p_g)}(g,t)$,
using cohort-share weights
$\hat w_g = N_g / N_{\text{treated}}$, uniform across post-treatment
event times within each cohort
(see Table~\ref{tab:weights} for sensitivity to alternative schemes).

\medskip\noindent
\textbf{Three reporting strategies.}
Three canonical order assignments provide a sensitivity check.
\emph{Strategy I} (Uniform Conservative): $p_g = 1$ for all
cohorts, applying the strictest common assumption uniformly.
\emph{Strategy II} (Uniform Feasible): choose target $p \leq \min_g m_g$, restrict
to $\mathcal{F}(p)$, same assumption for all included cohorts.
\emph{Strategy III} (Cohort Maximum): $p_g = m_g - 1$, the highest
feasible order for each cohort given its pre-treatment series,
ensuring at least one residual degree of freedom for the polynomial fit.
This strategy fully exploits the heterogeneous-order aggregation of
Theorem~\ref{thm:aggregation} and is the recommended default when
pre-treatment series are long and R$^2$ diagnostics support higher-order
fits across cohorts.  When pre-treatment series are short, the cohort-maximum order may be poorly identified; in such cases the sensitivity table and R$^2$ diagnostics should guide
the choice between Strategy II and Strategy III.\footnote{Reporting all three strategies provides a sensitivity check on the choice of cohort weighting scheme.}

\begin{remark}[Uniform versus cohort-heterogeneous orders in practice]
  \label{rem:uniform}
  Theorem~\ref{thm:aggregation} establishes identification when cohorts are identified under different feasible orders $p_g$. In practice, researchers may either let the sequential algorithm in Section~\ref{sec:orderselection} assign each cohort its own selected order, or impose a single uniform order $p^\star$ chosen from the pre-period diagnostics, which simplifies interpretation and comparison across cohorts. Either choice is covered by the heterogeneous-order aggregation in Theorem~\ref{thm:aggregation}.

  Cohort-heterogeneous orders are appropriate when there is strong prior reason to believe different cohorts face different pre-trend structures, or when later-treated cohorts have substantially more pre-treatment periods and the researcher wishes to exploit additional data available for those cohorts. In either case, transparency requires reporting which order was
  applied to each cohort.
\end{remark}

\subsection{Asymptotic Properties}

\begin{proposition}[Asymptotic Normality]
  \label{prop:asym}
  Under Assumptions~\ref{ass:na}, \ref{ass:ov},
  and \ref{ass:ptp} (for the relevant order per cohort),
  bounded second moments, independent sampling across units,
  and $\pi_g > 0$ for all $g$:
  \begin{enumerate}[label=(\roman*)]
    \item $\sqrt{N}(\widehat{\mathrm{ATT}}^{(p_g)}(g,t) - \ATTgt)
      \xrightarrowd \mathcal{N}(0, V_{g,t})$.
    \item $\sqrt{N}(\hat\theta - \theta)
      \xrightarrowd \mathcal{N}(0, V)$.
    \item $V$ is consistently estimated by the cluster bootstrap
      described below.
  \end{enumerate}
\end{proposition}

\begin{proof}
  The sample gap $\hat\gamma_{g,t}$ is a difference of sample
  means.
  By independence across units and bounded second moments:
  $\sqrt{N}(\hat\gamma_{g,t} - \gamma_{g,t})
  \xrightarrowd \mathcal{N}(0, \sigma^2_{g,t})$
  by the central limit theorem, where
  $\sigma^2_{g,t} = \Var(Y_{it}\mid G_i=g)/\pi_g
  + \Var(Y_{it}\mid G_i=\infty)/\pi_\infty$.

  The OLS polynomial coefficients $\hat{\mathbf{c}}_g$ are a linear
  function of the pre-treatment gaps $\{\hat\gamma_{g,t}\}_{t<g}$:
  $\hat{\mathbf{c}}_g = (\mathbf{V}^\top\mathbf{V})^{-1}
  \mathbf{V}^\top\hat{\boldsymbol{\gamma}}_g^{\mathrm{pre}}$,
  where $\mathbf{V}$ is the $(m_g\times p_g)$ Vandermonde matrix
  of pre-treatment time polynomials.
  By the delta method applied to the jointly normal vector of
  pre-treatment sample gaps,
  $\sqrt{N}(\hat{\mathbf{c}}_g - \mathbf{c}_g)
  \xrightarrowd \mathcal{N}(\mathbf{0}, \boldsymbol{\Sigma}_c)$
  where
  $\boldsymbol{\Sigma}_c
  = (\mathbf{V}^\top\mathbf{V})^{-1}\mathbf{V}^\top
  \boldsymbol{\Sigma}_\gamma^{\mathrm{pre}}
  \mathbf{V}(\mathbf{V}^\top\mathbf{V})^{-1}$
  and $\boldsymbol{\Sigma}_\gamma^{\mathrm{pre}}$ is the covariance
  matrix of the pre-treatment gaps.

  The counterfactual $\hat\gamma_{g,t}(0)
  = \mathbf{v}_t^\top\hat{\mathbf{c}}_g$,
  where $\mathbf{v}_t = (1,t,\ldots,t^{p_g-1})^\top$, is linear
  in $\hat{\mathbf{c}}_g$, hence also asymptotically normal by
  the delta method.
  The ATT estimate
  $\widehat{\mathrm{ATT}}^{(p_g)}(g,t)
  = \hat\gamma_{g,t} - \hat\gamma_{g,t}(0)$
  is the difference of two jointly asymptotically normal
  quantities, hence asymptotically normal.

  Part (ii): the aggregate $\hat\theta$ is a finite weighted sum
  of cohort-level ATT estimates.
  All cohorts share the same never-treated control group, so their
  ATT estimators are not independent: sampling variation in the
  common control group enters every
  $\widehat{\mathrm{ATT}}^{(p_g)}(g,t)$.
  The correct aggregate variance, as the limit of $N$ times
  the finite-sample covariance matrix, is:
  $V = \lim_{N\to\infty} N \cdot
  \sum_{g,t,g',t'} w_{g,t}w_{g',t'}
  \Cov(\widehat{\mathrm{ATT}}^{(p_g)}(g,t),
  \widehat{\mathrm{ATT}}^{(p_{g'})}(g',t'))$,
  which includes cross-cohort covariance terms arising through the
  shared control.
  Asymptotic normality of $\hat\theta$ nonetheless follows because
  $\hat\theta$ is a linear function of the jointly asymptotically
  normal vector of all gap estimates, and the delta method applies.
  The full variance $V$, including cross-cohort terms, is
  consistently estimated by the cluster bootstrap (Part iii).

  Part (iii): the cluster bootstrap (resampling entire units)
  consistently estimates the asymptotic variance under the
  stated regularity conditions, because it correctly replicates
  the joint sampling distribution of $(\hat\gamma_{g,t})_{g,t}$
  across cohorts; see \citet{callaway2021difference} for details
  of the argument in the analogous setting.\footnote{Note that serial correlation within units is  accommodated by the cluster bootstrap, which resamples entire units and thereby preserves within-unit  dependence across time. The analytical variance formula in Part (i) assumes independence across time periods within units and should not be used directly under serial correlation; inference should rely on the bootstrap throughout.}
\end{proof}

\begin{remark}[Small cohorts and finite-sample fragility]
  \label{rem:small_cohorts}
  Proposition~\ref{prop:asym} assumes each cohort fraction
  $\pi_g$ remains positive as $N \to \infty$.
  In practice, the delayed-expansion cohorts in the Medicaid
  application contain 3, 2, 1, and 2 states respectively.
  A percentile cluster bootstrap cannot establish a central
  limit theorem for a singleton treated cohort.
  Inferential results for cells involving the 2017 cohort
  (one state) and other small cohorts should be treated with
  caution; the reported estimates for these cohorts are best
  understood as descriptive comparisons rather than
  asymptotically valid inference.
\end{remark}

\textbf{Influence function.}
The influence function for
$\widehat{\mathrm{ATT}}^{(p_g)}(g,t)$ has two components.
Let $\psi_{it}^{\mathrm{post}}$ be the unit's contribution to
the post-treatment gap $\hat\gamma_{g,t}$, and let
$\psi_{it}^{\mathrm{pre}}$ capture its contribution to the
polynomial coefficients through the pre-treatment gaps.
The counterfactual evaluated at $t$ is linear in the
pre-treatment gaps via the Vandermonde projection:
$\hat\gamma_{g,t}(0)
= \mathbf{v}_t^\top(\mathbf{V}^\top\mathbf{V})^{-1}
\mathbf{V}^\top\hat{\boldsymbol{\gamma}}_g^{\mathrm{pre}}$.
The influence function for $\hat\gamma_{g,t}(0)$ is therefore a
linear combination of influence functions for the pre-treatment
gaps, with weights
$\mathbf{v}_t^\top(\mathbf{V}^\top\mathbf{V})^{-1}
\mathbf{V}^\top$.
The full influence function for
$\widehat{\mathrm{ATT}}^{(p_g)}(g,t)$
is $\psi_{it} = \psi_{it}^{\mathrm{post}} - \psi_{it}^{\mathrm{cf}}$,
where $\psi_{it}^{\mathrm{cf}}$ is the component arising from
polynomial estimation.
When $p_g=1$, $\psi_{it}^{\mathrm{cf}}$ is analogous to the
pre-treatment contribution in the Callaway--Sant'Anna influence
function, though the two differ when more than one pre-treatment
period is used: the present estimator averages all pre-treatment
gaps, while the Callaway--Sant'Anna construction uses the
relevant base period.

\textbf{Cluster bootstrap.}
For state-level panel data, all inference uses the cluster
bootstrap: resample entire states with replacement ($B$ draws),
compute $\widehat{\mathrm{ATT}}^{(p_g)}(g,t)$ on each draw,
and construct 95\,\% confidence intervals using the normal approximation
(estimate $\pm z_{0.025} \times$ bootstrap standard error).
Section~\ref{sec:simulation} shows that $B = 999$ draws yields
stable standard errors for datasets of the size studied here.

\paragraph{Dependence across cohorts.}
When states belong to the same geographical or institutional
cluster (e.g., Census division), their outcomes may be correlated
across cohorts.
The cluster bootstrap at the state level accounts for within-state
dependence across time but assumes independence across states.
In settings with strong cross-state dependence (e.g., common macro
shocks affecting all states simultaneously), the standard errors
may be understated.
One recommendation is to supplement with randomisation inference,
assigning treatment dates randomly to
never-treated states and checking whether the estimated ATT exceeds
the permutation distribution \citep{rambachan2023more}.

\section{Order Selection and Diagnostic Tests}
\label{sec:orderselection}

\subsection{Testable Restrictions}

Lemma~\ref{lem:poly} implies that when $m_g > p$, Parallel[$p$]
generates $m_g - p$ testable restrictions on pre-treatment data.
For cohort $g$, define:
\begin{equation}
  T_g(p)
  = \frac{N_g N_\infty}{N}
  \sum_{t=t_{\min}+p}^{g-1}
  \frac{(\difp{p}\hat\gamma_{g,t})^2}
       {\widehat{\mathrm{Var}}(\difp{p}\hat\gamma_{g,t})}.
  \label{eq:tstat}
\end{equation}
Under the null Parallel[$p$] and the additional assumption that
the higher-order differences $\difp{p}\hat\gamma_{g,t}$ are
asymptotically uncorrelated across $t$, $T_g(p) \xrightarrowd \chi^2(m_g-p)$.
This uncorrelatedness condition is \emph{not} automatically
satisfied even under serial independence of raw outcomes:
applying the $p$-th difference operator to an i.i.d.\ sequence
mechanically induces a moving-average dependence of order $p-1$
among the differenced terms.
In practice, $T_g(p)$ is therefore best treated as a descriptive
diagnostic rather than a formal test; researchers should
complement it with the pre-period $R^2$ and visual residual
inspection, and rely on the cluster bootstrap for inference.
Pooling across cohorts:
\begin{equation}
  T(p) = \sum_g T_g(p)
  \;\xrightarrowd\; \chi^2\!\Bigl(\sum_g(m_g-p)\Bigr)
  \quad \text{under Parallel[$p$]}.
  \label{eq:jtstat}
\end{equation}
In settings with serial correlation or strong cross-cohort
dependence, the chi-square approximation may be unreliable and
practitioners should rely on permutation-based critical values or
treat the statistic as a descriptive diagnostic rather than a formal
test. The cluster bootstrap used for inference does not depend on
these distributional assumptions.

\subsection{Sequential Algorithm}

Let $\alpha$ be the significance level. For each cohort $g$
the search runs up to the cohort-specific cap
$p_{\max,g} = m_g - 1$.

\begin{enumerate}[label=Step~\arabic*.]
  \item For each cohort $g$, test $T_g(1)$ at level $\alpha$.
    If not rejected, adopt $p_g=1$ for that cohort.
  \item If $T_g(1)$ rejected, test $T_g(2)$.
    If not rejected, adopt $p_g=2$.
  \item Continue until either the current test is not rejected
    or $p_g = p_{\max,g}$.
\end{enumerate}

\noindent The procedure is applied independently to each cohort,
allowing different cohorts to be identified under different selected
orders. The pre-period $R^2$ from the polynomial fit provides a
complementary diagnostic: low $R^2$ at a given order signals poor
polynomial fit regardless of the test outcome, which can occur when
pre-treatment series are short and the F-test has limited power.
In practice, $R^2$ and the sensitivity of the aggregate ATT across
orders should guide order selection alongside the formal test.

\begin{remark}[Polynomial misspecification diagnostic]
  \label{rem:poly_diag}
  Before applying the sequential test, I recommend inspecting the polynomial fit via the in-sample $R^2$ from the pre-period OLS regression. Low $R^2$ signals that the pre-treatment gap does not follow a polynomial of the chosen degree, potentially indicating structural breaks, seasonal patterns, or logistic growth. Section~\ref{sec:simulation} shows that under such misspecification, $R^2$ deteriorates noticeably, providing an early warning. In such cases, the polynomial extrapolation may not be reliable and 
researchers may find it useful to complement the point estimates with sensitivity bounds 
\citep{rambachan2023more} to assess robustness to departures from the polynomial structure.
\end{remark}

\begin{remark}[Extrapolation horizon and polynomial behaviour]

\label{rem:roth_poly}

A potential limitation of polynomial extrapolation is that
fitted polynomials can behave erratically when projected
beyond the support of the estimation data.
An analogous concern motivates the critique of high-degree
polynomial regression in the regression discontinuity
literature \citep{gelman2019high}, though the setting here
differs: the polynomial is fitted to pre-treatment gaps
in time rather than to outcomes near a discontinuity
threshold. Appendix~\ref{app:horizon} quantifies how performance
changes with the extrapolation horizon.
Variance rather than bias is the primary cost of longer horizons.
Applied researchers working with long post-treatment windows
may find it useful to complement DD[$p$] estimates with
the derivative-bounded sensitivity bounds of
\citet{rambachan2023more}, which do not rely on polynomial
extrapolation.

\end{remark}

\begin{remark}[Post-selection inference]
  I recommend reporting estimates for $p=1,2,3$ regardless of the sequential test outcome. The selected order is the primary specification; others are robustness checks.
  Section~\ref{sec:simulation} demonstrates that when the correct
  order is selected, bootstrap confidence intervals achieve
  approximately nominal coverage. When over-selection occurs (using $p=2$ when $p=1$ is true), coverage deteriorates, this is the cost of unnecessary flexibility. The sequential test minimises over-selection by adopting the
  lowest order the data do not reject.
\end{remark}

\section{Monte Carlo Evidence}
\label{sec:simulation}

\subsection{Main Simulation Design}

I evaluate DD[$p$] for $p \in \{1,2,3\}$ under three
data-generating processes.
The panel has $N=300$ units, $T=12$ periods, three cohorts
($g=5,7,9$) with 4, 6, and 8 pre-treatment periods, and
40\,\% never-treated.
True ATT $= 0.5$ throughout.
Results are based on 500 replications.

\textbf{DGP-1} (\emph{True Parallel[1]}).
$Y_{it}(\infty) = \alpha_i + \lambda_t + \varepsilon_{it}$,
$\alpha_i \iid \mathcal{N}(0,1)$, $\varepsilon_{it} \iid
\mathcal{N}(0,1)$, $\lambda_t$ a common trend.
DD[1] is the efficient estimator.

\textbf{DGP-2} (\emph{True Parallel[2], not Parallel[1]}).
$Y_{it}(\infty) = \alpha_i + \beta_g t + \lambda_t
+ \varepsilon_{it}$, $\beta_g \iid \mathcal{N}(0,0.5)$.
Cohort-specific linear trends violate Parallel[1].
DD[2] is the correctly specified estimator.

\textbf{DGP-3} (\emph{True Parallel[3]}).
$Y_{it}(\infty) = \alpha_i + \beta_g t + \gamma_g t^2
+ \lambda_t + \varepsilon_{it}$, $\gamma_g \iid
\mathcal{N}(0,0.15)$.
Cohort-specific quadratic trends.

Table~\ref{tab:simresults} reports bias and RMSE; Figure~\ref{fig:sim} displays them.

\begin{table}[H]
\centering
\caption{Monte Carlo: Bias and RMSE of DD[$p$] Estimators}
\label{tab:simresults}
\small
\begin{tabular}{lcccccc}
\toprule
 & \multicolumn{2}{c}{\textbf{DD[1]}}
 & \multicolumn{2}{c}{\textbf{DD[2]}}
 & \multicolumn{2}{c}{\textbf{DD[3]}} \\
\cmidrule(lr){2-3}\cmidrule(lr){4-5}\cmidrule(lr){6-7}
\textbf{DGP} & Bias & RMSE & Bias & RMSE & Bias & RMSE \\
\midrule
True Parallel[1]
  & \textbf{0.000} & \textbf{0.066}
  & 0.004 & 0.193 & 0.023 & 1.249 \\
True Parallel[2]
  & $-$0.001 & 0.240
  & \textbf{0.006} & \textbf{0.187}
  & 0.006 & 1.157 \\
True Parallel[3]
  & $-$0.025 & 0.937
  & $-$0.008 & 0.454
  & \textbf{0.019} & \textbf{1.244} \\
\bottomrule
\end{tabular}
\begin{minipage}{0.94\textwidth}\footnotesize\medskip
\textit{Notes}: Bold identifies the estimator matching the true DGP. Under DGP-2, DD[1] has elevated RMSE (0.240 vs.\ 0.187) without directional bias, because cohort-specific slopes $\beta_g$ have mean zero; the simulation isolates the \emph{variance} cost of misspecification. DD[3] has large RMSE in all settings because fitting a quadratic to four to eight pre-treatment observations is imprecise here; with longer pre-treatment series, RMSE would be expected to fall, consistent with the general guidance in Section~\ref{sec:est} that higher orders are most reliable when pre-treatment series are long and diagnostics support them.
\end{minipage}
\end{table}

\begin{figure}[H]
  \centering
  \includegraphics[width=0.96\textwidth]{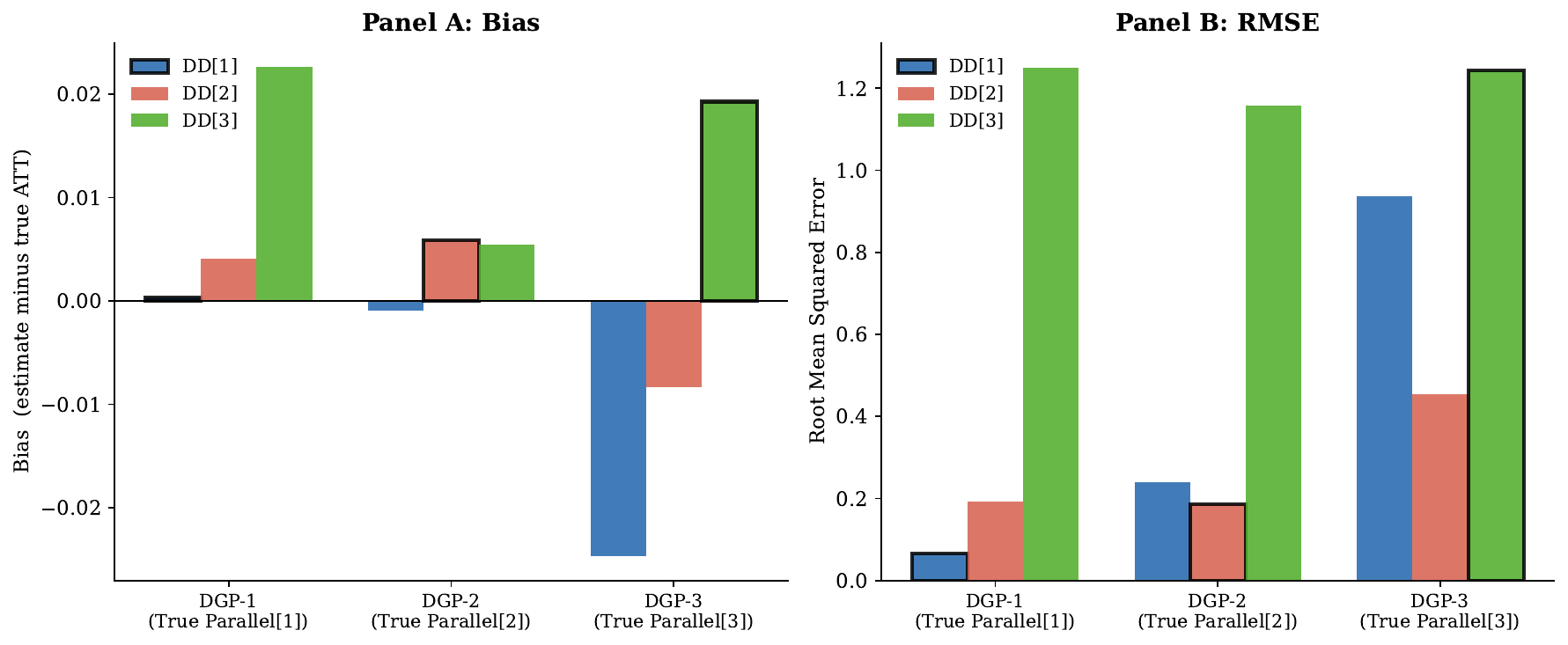}
  \caption{Simulation Study: Bias and RMSE by Estimator and DGP}
  \label{fig:sim}
  \begin{minipage}{0.94\textwidth}\footnotesize\medskip
    \textit{Notes}: Each bar shows bias (Panel A) or RMSE (Panel B)
    for one estimator under one DGP.
    Bold borders identify the estimator matching the true DGP.
    500 replications; see Table~\ref{tab:simresults} for details.
  \end{minipage}
\end{figure}

\subsection{Post-Selection Inference}
\label{subsec:postsel}

Table~\ref{tab:postselection} reports order-selection
frequencies and empirical coverage of 95\,\% confidence intervals
after the sequential order-selection procedure, based on 500
simulation draws with $B=99$ cluster bootstrap replications per draw.
Each row corresponds to a different data-generating process;
the selection frequencies show which order the sequential F-test
chooses under each DGP, and the coverage column reports how often
the resulting confidence interval contains the true ATT.

\begin{table}[H]
\centering
\caption{Bootstrap Coverage after Sequential Order Selection}
\label{tab:postselection}
\small
\begin{tabular}{lccccc}
\toprule
 & \multicolumn{3}{c}{\textbf{Selection frequency}} & & \\
\cmidrule(lr){2-4}
\textbf{DGP} & $p=1$ & $p=2$ & $p=3$ & \textbf{Coverage} & \textbf{Bias} \\
\midrule
True Parallel[1] (flat gap)          & 85.4\% & 11.2\% &  3.4\% & 89.8\% & $\phantom{-}$0.003 \\
True Parallel[2] (linear trend)      & 38.4\% & 52.0\% &  9.6\% & 93.6\% & $-$0.002 \\
True Parallel[3] (quadratic trend)   &  9.8\% & 59.2\% & 31.0\% & 93.8\% & $-$0.099 \\
Borderline (very weak linear trend)  & 86.8\% & 10.4\% &  2.8\% & 89.2\% & $-$0.019 \\
\bottomrule
\end{tabular}
\begin{minipage}{0.96\textwidth}\footnotesize\medskip
\textit{Notes}: 500 replications, $N=300$, $T=12$, true ATT $=0.5$.
Three cohorts at $g=5,7,9$.
Coverage is the empirical frequency with which the 95\% cluster
bootstrap CI (B=99) contains the true ATT after the sequential
order-selection step.
Under DGP-1 and the borderline DGP, coverage is slightly below
the nominal 95\%, reflecting finite-sample bootstrap approximation
at $B=99$; bias is negligible in both cases.
The empirical application uses $B=999$.
Under DGP-2, the sequential test selects the correct order ($p=2$)
in 52\% of draws; under-selection to $p=1$ accounts for the
remainder and slightly widens the coverage band.
Under DGP-3, the test stops at $p=2$ in 59\% of draws, reflecting
the difficulty of detecting quadratic structure at $N=300$.
\end{minipage}
\end{table}

\noindent
Selection frequencies are consistent with the test's design:
it selects the lowest order whose pre-treatment implications the
data do not reject, defaulting toward parsimony.
Coverage after selection is near-nominal for DGP-2 and DGP-3
(93.6\% and 93.8\%) and slightly below nominal for the flat and
borderline DGPs (89.8\% and 89.2\%), where finite-sample
bootstrap imprecision at $B=99$ accounts for the shortfall.

\subsection{Serial Correlation and Placebo}

Table~\ref{tab:robustness} reports two additional robustness checks.
Panel A shows that DD[2] maintains near-nominal coverage under
AR(1) serial correlation with persistence $\rho$ up to 0.7,
ranging from 91.4\% at $\rho=0$ to 94.4\% at $\rho=0.7$.
This reflects that the estimator operates on cohort-level averages,
which attenuate within-unit serial correlation.
Panel B shows a placebo false-positive rate of 7.2\%, modestly
above the nominal 5\%, consistent with the finite-sample bootstrap
behaviour at $B=99$ documented in Table~\ref{tab:postselection};
the mean estimated ATT under the null is essentially zero,
confirming the absence of systematic bias.

\begin{table}[H]
\centering
\caption{Robustness: Serial Correlation (Panel A) and
  Placebo (Panel B)}
\label{tab:robustness}
\small
\begin{tabular}{lcccc}
\toprule
\multicolumn{5}{l}{\textbf{Panel A: AR(1) serial correlation},
  true Parallel[2], DD[2], 95\% CI} \\
\midrule
AR(1) persistence ($\rho$) & 0.0 & 0.3 & 0.5 & 0.7 \\
Coverage & 91.4\% & 92.2\% & 94.0\% & 94.4\% \\
Bias & $\phantom{-}$0.023 & $-$0.019 & $-$0.009 & $\phantom{-}$0.002 \\
\midrule
\multicolumn{5}{l}{\textbf{Panel B: Placebo test},
  true ATT $= 0$, DGP-2 structure, DD[2], 500 replications} \\
\midrule
False positive rate (nominal 5\%) & \multicolumn{4}{l}{7.2\%} \\
Mean ATT under $H_0$ & \multicolumn{4}{l}{$-$0.008} \\
\bottomrule
\end{tabular}
\begin{minipage}{0.92\textwidth}\footnotesize\medskip
\textit{Notes}: 500 replications per scenario, $N=300$, $T=12$,
$B=99$ bootstrap draws.
Panel A shows that the cluster bootstrap maintains near-nominal
coverage under AR(1) serial correlation up to $\rho=0.7$,
because the estimator operates on cohort-level averages which
attenuate within-unit serial dependence.
Panel B false positive rate of 7.2\% is modestly above the
nominal 5\%, consistent with the slight under-coverage observed
in Table~\ref{tab:postselection} at $B=99$; the mean ATT under
$H_0$ is essentially zero, confirming the absence of systematic bias.
\end{minipage}
\end{table}

\subsection{Failure Modes and the Polynomial Diagnostic}

The simulations above calibrate DGPs to exact polynomial trends.
When the pre-treatment gap follows a non-polynomial trajectory, such as a structural break or logistic growth, Parallel[$p$] is misspecified regardless of the chosen order.

The practical diagnostic is the pre-period $R^2$ from the OLS
polynomial fit: under a genuine polynomial trend $R^2$ is high;
under structural breaks or other non-polynomial dynamics $R^2$
deteriorates noticeably, providing an early warning before the
sequential test is applied.
See Remark~\ref{rem:poly_diag}.
When the $R^2$ diagnostic signals misspecification, sensitivity bounds \citep{rambachan2023more} may be more appropriate.
Table~\ref{tab:bootstab} shows how the cluster bootstrap standard
error for the aggregate DD[2] ATT stabilises as $B$ increases,
using the Medicaid expansion application data.

\begin{table}[H]
\centering
\caption{Bootstrap Standard Error Stability as $B$ Increases:
  Medicaid Application}
\label{tab:bootstab}
\begin{tabular}{lrrrrrr}
\toprule
$B$ & 25 & 50 & 100 & 200 & 500 & 999 \\
\midrule
SE(DD[2]) & 0.0118 & 0.0103 & 0.0099 & 0.0100 & 0.0096 & 0.0096 \\
\bottomrule
\end{tabular}
\begin{minipage}{0.82\textwidth}\footnotesize\medskip
\textit{Notes}: Medicaid expansion application, 46 states, DD[2],
aggregate ATT.
$B=999$ used throughout; standard errors stabilise by $B=100$.
\end{minipage}
\end{table}

\section{Empirical Application: Medicaid Expansion}
\label{sec:empirical}

\subsection{Setting and Data}

I use a state-year panel covering 46 states observed over 
2008--2019, yielding 552 observations.\footnote{Data are 
drawn from the Mixtape Sessions Advanced DiD repository, made publicly 
available by Scott Cunningham and accessible at 
\url{https://raw.githubusercontent.com/Mixtape-Sessions/Advanced-DID/main/Exercises/Data/ehec_data.dta}. 
The underlying microdata source is the American Community 
Survey (ACS). } The outcome (\texttt{dins}) 
is the share of low-income childless adults with health 
insurance in each state and year, derived from the American Community Survey (ACS). 
Treatment is Medicaid expansion under the Affordable Care 
Act of 2010, which gave states the option to expand 
Medicaid eligibility beginning in 2014. Five expansion 
cohorts are present: 2014 (22 states, $m_g=6$), 2015 (3, 
$m_g=7$), 2016 (2, $m_g=8$), 2017 (1, $m_g=9$), and 
2019 (2, $m_g=11$). Sixteen states never expanded and 
serve as the never-treated comparison group throughout. 
Population weights are normalised to mean one to ensure 
interpretable test statistics.\footnote{Raw weights 
average 647,395 per state-year, inflating chi-squared 
statistics proportionally to the sum of weights rather 
than the number of observations when used un-normalised. 
Normalising preserves relative weighting across states.}

\subsection{Pre-Trend Diagnostics}

The Callaway--Sant'Anna aggregate ATT under Parallel[1], using
state population weights, is
$\hat\theta_{\mathrm{CS}} = 0.075$ (95\,\% CI: [0.051, 0.099],
$p < 0.001$).
The joint pre-trend test yields $\chi^2(36) = 65{,}921$,
$p < 0.001$, decisively rejecting standard parallel
trends.\footnote{The large chi-squared reflects importance-weight scaling in \texttt{csdid}; see footnote~1. The qualitative rejection is unaffected.}
The pre-treatment gap for the 2014 cohort rises steadily from
0.032 in 2008 to 0.042 in 2013, consistent with a linear trend.
I select $p^\star = 2$ as the primary specification
based on pre-period $R^2$ diagnostics and sensitivity stability
across orders (Table~\ref{tab:orderselection}).
For transparency, the fully data-driven sequential procedure
(\texttt{anddp} with \texttt{maxorder(4)}) selects $p_g = 1$ for the
2014 and 2016 cohorts, whose flat gap is not rejected (F-test
$p$-values 0.117 and 0.273, reflecting the limited power of the test
with 6--8 pre-treatment periods), and $p_g = 2$ for the 2015, 2017,
and 2019 cohorts, yielding an aggregate of 0.074 (SE 0.0096, 95\%
CI [0.055, 0.093]) under the OLS-mean baseline convention of
Section~\ref{sec:est}. I nonetheless adopt the uniform
$p^\star = 2$ as primary: the 2014 gap drifts monotonically upward
across all six pre-treatment years, the linear fit raises the
pre-period $R^2$ from 0.00 to 0.50, and the aggregate is stable
across orders (0.071, 0.064, and 0.060 at $p = 1, 2, 3$ under the
OLS-mean convention).

\begin{figure}[H]
  \centering
  \includegraphics[width=0.82\textwidth]{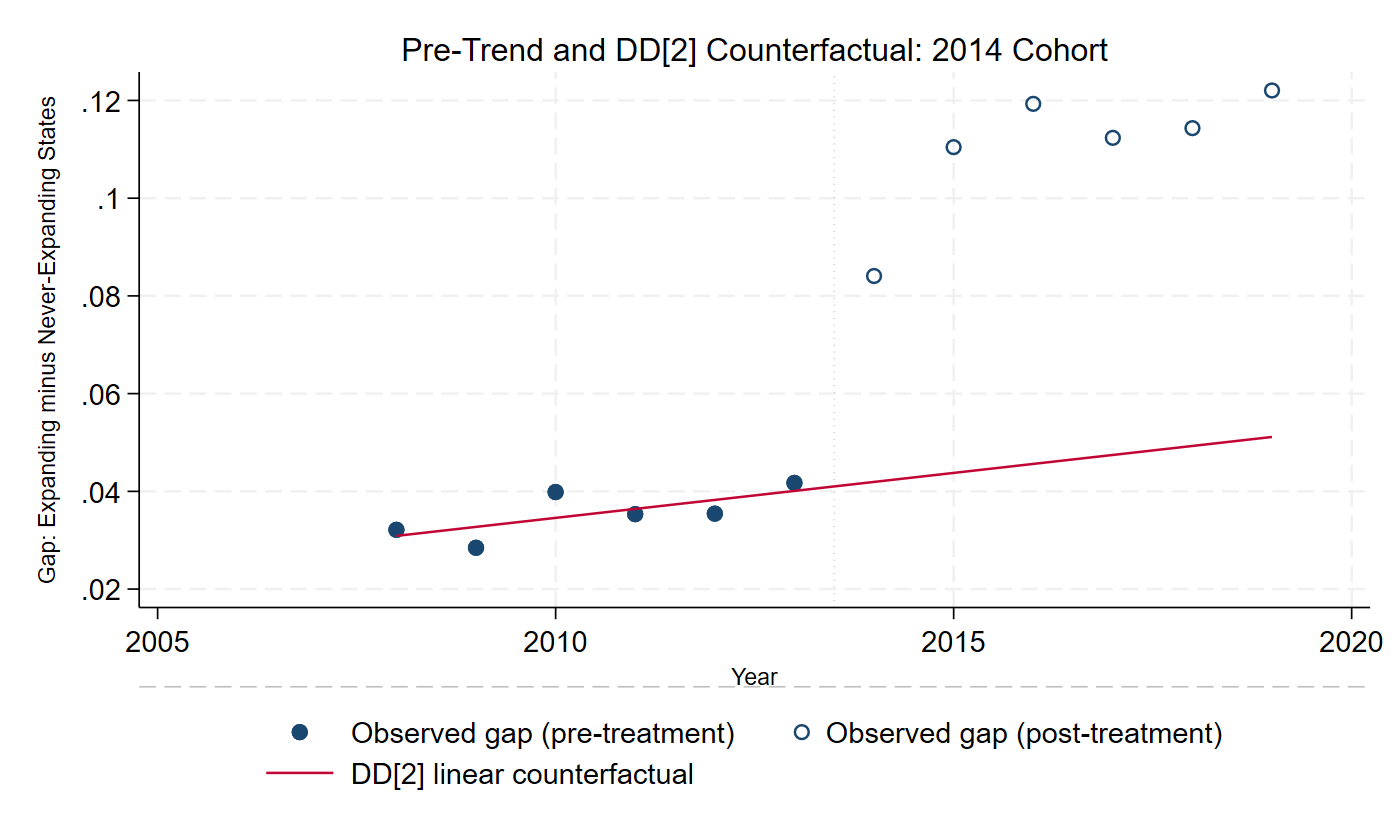}
  \caption{Pre-Treatment Gap and DD[2] Linear Counterfactual:
    2014 Expansion Cohort}
  \label{fig:gap}
  \begin{minipage}{0.82\textwidth}\footnotesize\medskip
    \textit{Notes}: Solid circles show the observed gap
    $\hat\gamma_{g,t}$ between the 2014 expansion cohort and
    never-treated states in each pre-treatment year (2008--2013).
    The red line is a linear trend fitted to the six pre-treatment
    gap observations by OLS and extended into the post-treatment
    period as the DD[2] counterfactual.
    The steady upward drift in pre-treatment observations confirms
    that standard parallel trends (which requires a flat gap) is
    violated, while the approximately linear pattern ($R^2 = 0.50$ for the
    linear fit) is consistent
    with Parallel[2] as the identifying assumption.

  \end{minipage}
\end{figure}

\begin{table}[H]
\centering
\caption{Pre-Trend Diagnostics and Order Selection}
\label{tab:orderselection}
\begin{tabular}{lcc}
\toprule
Diagnostic & Result & Decision \\
\midrule
CS aggregate ATT (weighted) & 0.075 [0.051, 0.099]$^{***}$ & Reference \\
Joint pre-trend test ($\chi^2$, 36 d.f.) & 65{,}921$^{***}$ & Reject $p=1$ \\
Pre-period polynomial fit ($p=2$) & $R^2$: 0.50, 0.87, 0.20, 0.54, 0.46 & Consistent with Parallel[2] for most cohorts \\
Selected order & $p^\star = 2$ & Primary specification \\
\bottomrule
\end{tabular}
\begin{minipage}{0.88\textwidth}\footnotesize\medskip
\textit{Notes}: $^{***}p<0.001$. $R^2$ values are the per-cohort fits of the linear ($p=2$) pre-treatment polynomial, listed in cohort order (2014, 2015, 2016, 2017, 2019). The sequential F-test is an additional diagnostic; however, when pre-treatment series are short, $R^2$ and sensitivity across orders are more reliable
guides. Pre-period implication of Parallel[2] not rejected for all cohorts.
\end{minipage}
\end{table}

\begin{table}[H]
\centering
\caption{Cohort-Specific Pre-Periods and Feasible Orders,
  Medicaid Expansion}
\label{tab:cohortorders}
\begin{tabular}{lcccc}
\toprule
Cohort & States ($N_g$) & Pre-periods ($m_g$) & Max feasible $p$ & Applied $p^\star$ \\
\midrule
2014 & 22 & 6 & 5 & 2 \\
2015 & 3  & 7 & 6 & 2 \\
2016 & 2  & 8 & 7 & 2 \\
2017 & 1  & 9 & 8 & 2 \\
2019 & 2  & 11 & 10 & 2 \\
Never treated & 16 & --- & --- & --- \\
\bottomrule
\end{tabular}
\begin{minipage}{0.92\textwidth}\footnotesize\medskip
\textit{Notes}: Order $p^\star=2$ is applied uniformly across cohorts, supported by pre-period $R^2$ diagnostics and stability of estimates across orders (Table~\ref{tab:orderselection}).
No cohorts are excluded under the Uniform Feasible strategy at $p=2$.
The 2017 cohort (one state) has the largest feasible order but smallest sample; inference for this cohort should be treated with caution. The companion Stata command \texttt{anddp} offers a wild cluster bootstrap option (\texttt{wcbootstrap}) for more reliable inference in such settings.
\end{minipage}
\end{table}

\subsection{Main Results}

Table~\ref{tab:main} reports DD[1], DD[2], and DD[3] estimates
at event times $\tau = 0,\ldots,4$ with 95\,\% confidence intervals
from the state-level cluster bootstrap ($B=999$).
Figure \ref{fig:comparison} displays the event-study comparison between DD[1] and the primary DD[2] specification. Figure \ref{fig:comparison1} extends this to include DD[3] as a robustness check.\footnote{The DD[1] estimates in Table~\ref{tab:main} use the last pre-treatment gap as the flat counterfactual, following the \citet{callaway2021difference} convention. The unweighted CS simple aggregate (0.068) is close to the DD[1] cohort-weighted average (0.065), since both use the same identifying assumption and the same baseline convention. The weighted CS estimate (0.075) differs because it uses state population weights; the DD[$p$] estimates in this table use equal cohort-share weights. The companion command \texttt{anddp} implements the OLS-mean baseline of Section~\ref{sec:est} for DD[1], averaging all pre-treatment gaps, which yields 0.071 in this application; the last-gap variant is reported in the table for direct comparability with CS.}

The findings show that \emph{first}, all three estimators find a positive, significant, and growing effect. Insurance coverage rose by approximately 4--5 percentage points in the year of expansion and 7--8 points four years later. This qualitative conclusion is robust across identification strategies.

\emph{Second}, the DD[1] cohort-weighted average (0.065) is nearly identical to the unweighted Callaway--Sant'Anna simple aggregate (0.068), confirming that DD[1] as implemented here --- using the last pre-treatment gap as the flat counterfactual --- reproduces the CS estimator closely. The gap between the weighted CS estimate (0.075) and these figures reflects population weighting, not the identifying assumption.

\emph{Third}, DD[2] consistently exceeds DD[1] at every event time, with differences of 0.4--1.0 percentage points (10--22 percent of the DD[1] estimate). This reflects that later-treated cohorts (2015--2019) were on a modest downward trajectory relative to never-expanding states before their expansion. The standard estimator, which assumes a flat counterfactual, underestimates the treatment effect for these cohorts. DD[2] corrects for this by projecting the downward pre-trend forward.

\emph{Fourth}, the confidence intervals of DD[1] and DD[2] overlap substantially at every event time: I cannot reject the null that the two estimators produce the same treatment effect. The contribution of DD[2] in this application is therefore not a reversal of the standard finding but a more credible quantification of it. The same qualitative conclusion,  Medicaid expansion increased insurance coverage, is supported by an assumption the pre-treatment data do not reject, rather than one they reject.

\begin{table}[H]
\centering
\caption{DD[$p$] Estimates by Event Time: Medicaid Expansion}
\label{tab:main}
\begin{tabular}{lcccccc}
\toprule
 & \multicolumn{2}{c}{\textbf{DD[1]}}
 & \multicolumn{2}{c}{\textbf{DD[2]}}
 & \multicolumn{2}{c}{\textbf{DD[3]}} \\
\cmidrule(lr){2-3}\cmidrule(lr){4-5}\cmidrule(lr){6-7}
$\tau$ & Estimate & 95\% CI & Estimate & 95\% CI & Estimate & 95\% CI \\
\midrule
0 & 0.042 & [0.029, 0.056] & 0.052 & [0.036, 0.067] & 0.045 & [0.029, 0.060] \\
1 & 0.057 & [0.043, 0.071] & 0.064 & [0.045, 0.082] & 0.055 & [0.030, 0.079] \\
2 & 0.068 & [0.056, 0.080] & 0.074 & [0.061, 0.087] & 0.061 & [0.015, 0.107] \\
3 & 0.073 & [0.058, 0.087] & 0.077 & [0.060, 0.095] & 0.077 & [0.009, 0.145] \\
4 & 0.073 & [0.054, 0.092] & 0.084 & [0.057, 0.111] & 0.074 & [$-$0.013, 0.160] \\
\midrule
Simple avg.    & 0.060 & & 0.067 & & --- & \\
Cohort-wtd.\ avg. & 0.065 & & 0.064 & & --- & \\
CS (unweighted) & \multicolumn{2}{c}{0.068 [0.053, 0.083]} & & & & \\
CS (weighted)  & \multicolumn{2}{c}{0.075 [0.051, 0.099]} & & & & \\
\bottomrule
\end{tabular}
\begin{minipage}{0.96\textwidth}\footnotesize\medskip
\textit{Notes}: Never-treated states as comparison group.
95\% CIs from state-level cluster bootstrap ($B=999$).
$\tau=0$: year of expansion. Outcome: \texttt{dins}.
Cohorts: 2014 (22 states), 2015 (3), 2016 (2), 2017 (1), 2019 (2).
$\tau=5$ excluded (only 2014 cohort; bootstrap SE unreliable).
CS aggregate (weighted) from \citet{callaway2021difference} under
Parallel[1], $p<0.001$; see also the unweighted figure (0.068)
reported in the table body and discussed in the footnote above.
DD[3] confidence intervals are wide, consistent with the large
RMSE of DD[3] documented in Table~\ref{tab:simresults} for short
pre-treatment series; DD[3] is reported alongside DD[1] and DD[2]
as a sensitivity check in this application.
\end{minipage}
\end{table}

\begin{figure}[H]
  \centering
  \includegraphics[width=0.96\textwidth]{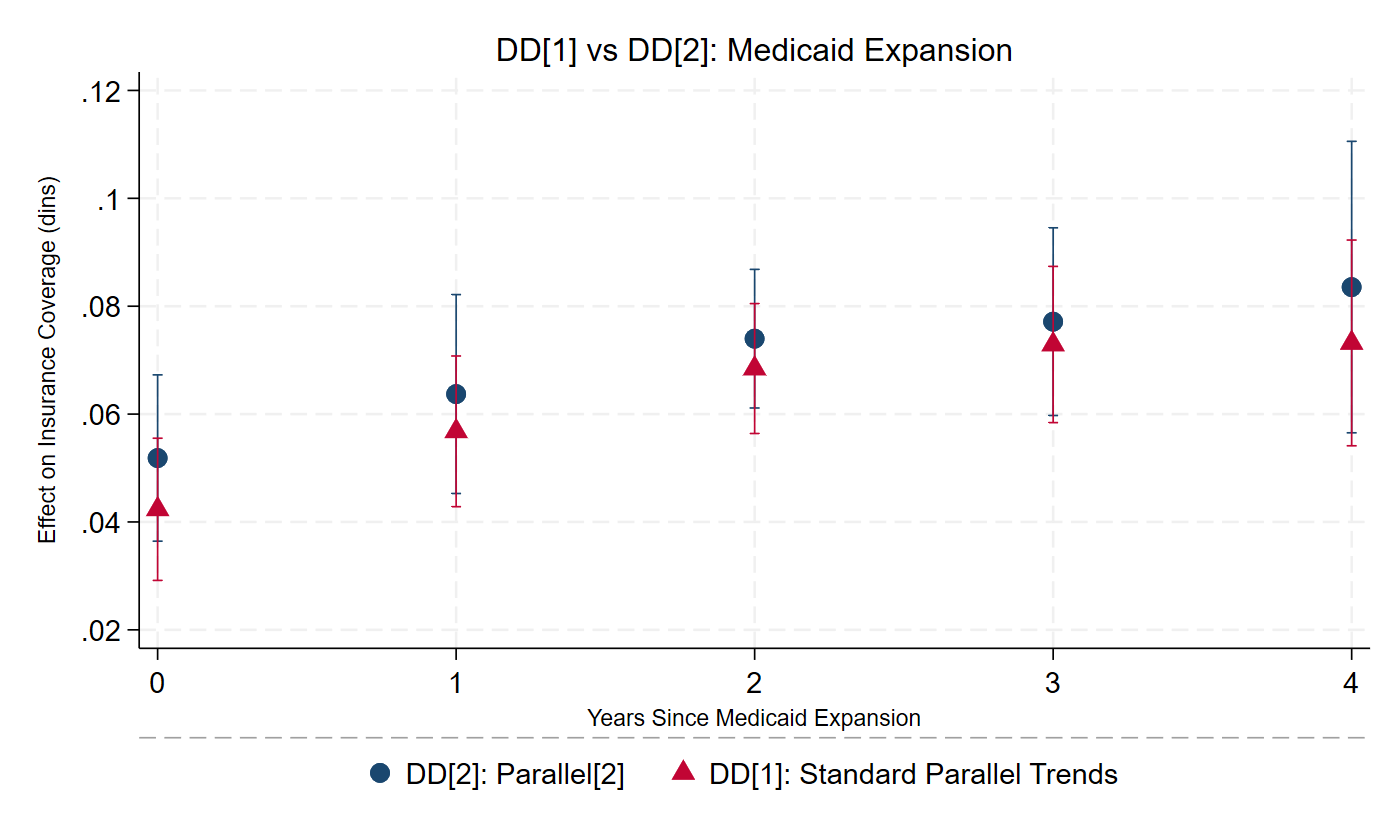}
  \caption{Event-Study Comparison: DD[1] and DD[2],
    Medicaid Expansion}
  \label{fig:comparison}
  \begin{minipage}{0.94\textwidth}\footnotesize\medskip
    \textit{Notes}: Never-treated states as comparison group.
    95\% CIs from cluster bootstrap ($B=999$).
    $\tau=5$ excluded.
  \end{minipage}
\end{figure}

\begin{figure}[H]
  \centering
  \includegraphics[width=0.96\textwidth]{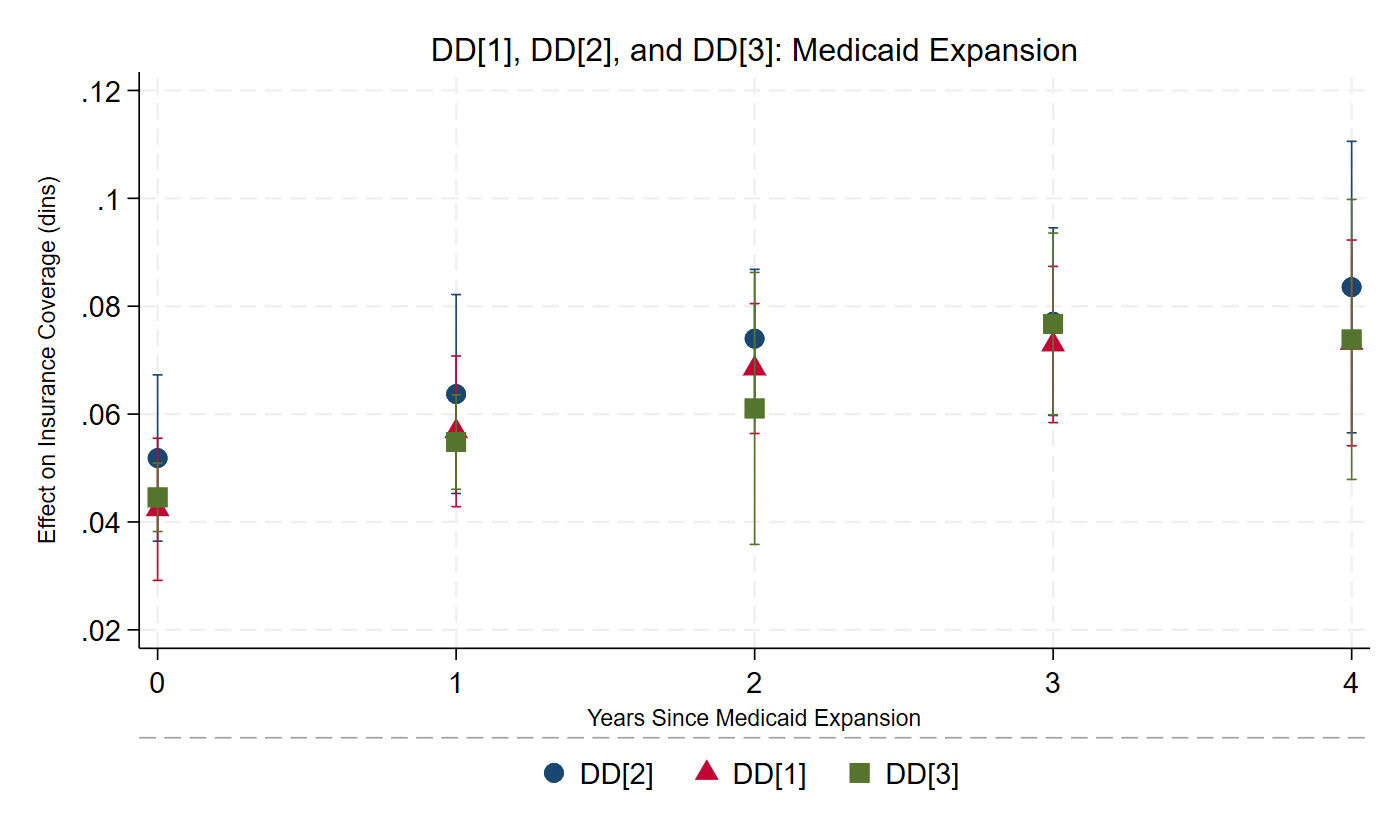}
  \caption{Event-Study Comparison: DD[1], DD[2] and DD[3],
    Medicaid Expansion}
  \label{fig:comparison1}
  \begin{minipage}{0.94\textwidth}\footnotesize\medskip
    \textit{Notes}: Never-treated states as comparison group.
    95\% CIs from cluster bootstrap ($B=999$).
    $\tau=5$ excluded.
  \end{minipage}
\end{figure}

Table~\ref{tab:summary} summarises estimates across all three specifications and records which pass the pre-trend test.

\begin{table}[H]
\centering
\caption{Summary of DD[$p$] Estimates and Diagnostic Tests,
  Medicaid Expansion}
\label{tab:summary}
\begin{tabular}{lccc}
\toprule
 & \textbf{DD[1]} & \textbf{DD[2]} & \textbf{DD[3]} \\
\midrule
Event time $\tau=0$ & 0.042 & 0.052 & 0.045 \\
Event time $\tau=2$ & 0.068 & 0.074 & 0.061 \\
Event time $\tau=4$ & 0.073 & 0.084 & 0.074 \\
Simple average      & 0.060 & 0.067 & --- \\
Cohort-weighted avg.& 0.065 & 0.064 & --- \\
CS aggregate ATT (weighted) & \multicolumn{3}{c}{0.075 [0.051, 0.099]} \\
\midrule
Pre-trend test      & Rejected$^{***}$ & Consistent$^{(a)}$ & Consistent$^{(a)}$ \\
Pre-period implication not rejected & No & Yes & Yes \\
\bottomrule
\end{tabular}
\begin{minipage}{0.92\textwidth}\footnotesize\medskip
\textit{Notes}: $^{***}$$\chi^2(36)=65{,}921$, $p<0.001$.
All ATT estimates significantly different from zero at the 1\% level.
CS aggregate (weighted) from \citet{callaway2021difference} under Parallel[1].
$^{(a)}$Based on the pre-period $R^2$ diagnostics reported in
Table~\ref{tab:orderselection}, not a separately computed formal test
statistic at these orders.
\end{minipage}
\end{table}

\subsection{Weighting Sensitivity}
\label{subsec:weightsens}

Table~\ref{tab:weights} shows the aggregate ATT under three
weighting strategies.
Results are insensitive to weighting choice: DD[2] ranges from
0.064 to 0.067 and DD[1] from 0.060 to 0.065.
The cohort-level results in the penultimate row reveal that the
direction of the DD[2] correction varies by cohort: for the large
2014 cohort, DD[2] is slightly smaller than DD[1] (its pre-trend
was upward), while for later cohorts DD[2] exceeds DD[1] (their
pre-trends were downward).
This heterogeneity in direction is precisely what motivates
cohort-specific counterfactuals.

\begin{table}[H]
\centering
\caption{Weighting Sensitivity: Aggregate DD[$p$] ATT}
\label{tab:weights}
\begin{tabular}{lcc}
\toprule
Weighting strategy & DD[1] & DD[2] \\
\midrule
I.\ Equal cell weight & 0.060 & 0.067 \\
II.\ Cohort-share (proportional to $N_g$) & 0.065 & 0.064 \\
III.\ Event-time weighted ($w\propto\tau+1$) & 0.066 & 0.073 \\
\bottomrule
\end{tabular}
\begin{minipage}{0.72\textwidth}\footnotesize\medskip
\textit{Notes}: Three canonical aggregation strategies.
The Callaway--Sant'Anna aggregate (0.075) uses population
weights and a different aggregation scheme than the equal-weight
and cohort-share schemes reported here.
\end{minipage}
\end{table}

\subsection{Relation to Sensitivity Bounds}
\label{subsec:bounds}

A natural benchmark is the sensitivity bounds approach of \citet{rambachan2023more}, which restricts the magnitude of parallel trends violations and reports how estimated effects vary within that restriction. Their approach delivers valid intervals whose width reflects the assumed maximum violation magnitude. Notably, \citet{rambachan2023more} use the same dataset as an illustrative example in their software package \texttt{HonestDiD}. Their smoothness restriction at $M=0$ imposes that the post-treatment counterfactual gap follows a linear extrapolation
of the pre-existing trend which is conceptually equivalent to DD[2] at Parallel[2].

For $M > 0$, their approach allows deviations from linearity up to a bound of size $M$, providing a natural complement to the point estimates reported here: DD[2] delivers a point estimate under the polynomial structure, while \citet{rambachan2023more} deliver robust intervals that remain valid even if the linear structure is imperfect. DD[2] complements this analysis by showing that if the specific polynomial structure is credible,  as the pre-treatment $R^2$ values and visual evidence suggest for most cohorts, one can achieve point identification without widening the bounds. The two approaches are not competitors: if the $R^2$ diagnostic signals misspecification, the \citet{rambachan2023more} bounds
remain the appropriate tool.

\section{Conclusion}
\label{sec:conc}

When pre-treatment event studies reject standard parallel trends, the applied researcher faces a difficult choice: impose a maintained assumption the data contradict, or report sensitivity bounds that are valid but do not deliver a point estimate. This paper offers a third option: replacing the flatness requirement with a structured polynomial extrapolation whose pre-treatment implications are directly checkable.

The central theoretical contribution is Theorem~\ref{thm:aggregation}, an aggregation result allowing cohorts to be identified under different feasible orders, arising naturally from variation in pre-treatment period counts across staggered adoption cohorts. This problem has no analogue in two-group settings; \citet{egami2023using} sketch a higher-order generalization for staggered designs in their appendix but do not develop an aggregation result, asymptotic theory, or applied procedure for it.

Monte Carlo evidence shows near-nominal coverage when the correct order is selected, robustness to AR(1) serial correlation, and a placebo false positive rate modestly above the nominal 5\,\%, consistent with finite-sample bootstrap behaviour at $B=99$.
Applied to Medicaid expansion, the approach yields estimates resting on an assumption the pre-treatment data do not reject. The Medicaid application uses annual data with 6-11 pre-treatment periods per cohort, representing a conservative environment for the polynomial fit. In administrative panel data at monthly or quarterly frequency, increasingly common in healthcare, labour market, and fiscal policy research, even two to three years of pre-treatment history yields 24-60 observations per cohort, giving the sequential F-test substantially greater power and making DD[3] and higher orders reliably estimable.

Two directions merit future work. \emph{First}, a formal uniform validity result for the complete selection-estimation pipeline would strengthen the inferential foundations. The extended version of \citet{roth2022pretest} on post-selection inference provides a starting point, though adapting those results to the sequential polynomial order selection here requires additional work. \emph{Second}, replacing the polynomial order restriction with derivative-bounded smoothness restrictions \citep{rambachan2023more} would extend the framework to settings where the pre-treatment gap does not follow a polynomial, at the cost of point identification.

\newpage
\bibliographystyle{plainnat}
\bibliography{references}

\appendix

\section{Regularity Conditions}
\label{app:reg}

The following conditions are maintained throughout.

\begin{enumerate}[label=A\arabic*.]
  \item \emph{Independence.}
    Units are mutually independent.
    Treatment assignment $G_i$ is independent of
    $\{Y_{it}(\infty), Y_{it}(g)\}_{t,g}$ conditional on
    group membership (implied by the potential outcomes
    structure).

  \item \emph{Moment conditions.}
    $\E[Y_{it}^2(\infty)] < \infty$ for all $i$, $t$.
    Cohort-level variances
    $\Var(Y_{it}\mid G_i=g) < \infty$ and
    $\Var(Y_{it}\mid G_i=\infty) < \infty$ for all $g$, $t$.

  \item[A2$'$.] \emph{Finite temporal covariance.}
    $|\Cov(Y_{it}(\infty), Y_{is}(\infty)\mid G_i=\infty)| < \infty$
    for all $t, s$.
    This condition is needed because the same never-treated units
    appear in both pre-treatment gaps (used to fit the polynomial)
    and post-treatment gaps, generating a non-zero covariance
    between the counterfactual and the post-treatment term.

  \item \emph{Non-degeneracy.}
    $\pi_g = \lim_{N\to\infty} N_g/N > 0$ for all $g$.
    $\pi_\infty = \lim_{N\to\infty} N_\infty/N > 0$.

  \item \emph{Pre-period richness.}
    $m_g \geq p_g$ for each cohort $g$.
    The Vandermonde matrix $\mathbf{V}$ of pre-treatment time
    polynomials has full column rank ($p_g$ distinct pre-treatment
    periods are observed for each cohort).

  \item \emph{Bounded time.}
    $T < \infty$ is fixed as $N \to \infty$.
\end{enumerate}

Under A1--A5 and A2$'$ and the identifying
Assumptions~\ref{ass:na}--\ref{ass:ptp},
the convergence results of Section~\ref{sec:est} hold with
$\sqrt{N}$ rates.
Asymptotic normality of cohort-level and aggregate ATT estimators
follows from the joint CLT applied to all gap estimates
simultaneously; the cluster bootstrap consistently estimates the
full variance, including cross-cohort covariance terms arising
from the shared never-treated control group (see
Appendix~\ref{app:proof_asym}).

\section{Full Proof of Proposition~\ref{prop:asym}}
\label{app:proof_asym}

\noindent\textbf{Part (i): Cohort-level ATT.}

Write
$\widehat{\mathrm{ATT}}^{(p_g)}(g,t)
= \hat\gamma_{g,t} - \hat\gamma_{g,t}(0)$.

\noindent\emph{Step 1: Asymptotic normality of $\hat\gamma_{g,t}$.}
Under A1--A3:
\[
  \sqrt{N}\bigl(\hat\gamma_{g,t} - \gamma_{g,t}\bigr)
  = \sqrt{N}\Bigl(
    \bar Y_{g,t} - \E[Y_{it}\mid G_i=g]
    - (\bar Y_{\infty,t} - \E[Y_{it}\mid G_i=\infty])
  \Bigr).
\]
The two terms are sample means from independent subsamples of
sizes $N_g = \pi_g N$ and $N_\infty = \pi_\infty N$.
By the CLT and A2:
$\sqrt{N}(\bar Y_{g,t} - \E[Y_{it}\mid G_i=g])
\xrightarrowd \mathcal{N}(0, \Var(Y_{it}\mid G_i=g)/\pi_g)$
and similarly for $\bar Y_{\infty,t}$.
By independence:
\[
  \sqrt{N}(\hat\gamma_{g,t} - \gamma_{g,t})
  \xrightarrowd \mathcal{N}(0, \sigma^2_{g,t}),
  \quad
  \sigma^2_{g,t}
  = \frac{\Var(Y_{it}\mid G_i=g)}{\pi_g}
  + \frac{\Var(Y_{it}\mid G_i=\infty)}{\pi_\infty}.
\]

\noindent\emph{Step 2: Asymptotic normality of $\hat\gamma_{g,t}(0)$.}
Stack the pre-treatment gaps into the vector
$\hat{\boldsymbol{\gamma}}_g^{\mathrm{pre}}
= (\hat\gamma_{g,t_{\min}},\ldots,\hat\gamma_{g,g-1})^\top
\in\mathbb{R}^{m_g}$.
By Step~1 applied jointly across all pre-treatment periods
(A1 ensures independence across units; within-unit temporal
dependence is captured through the covariance matrix
$\boldsymbol{\Sigma}_g$):
\[
  \sqrt{N}\bigl(
  \hat{\boldsymbol{\gamma}}_g^{\mathrm{pre}}
  - \boldsymbol{\gamma}_g^{\mathrm{pre}}
  \bigr)
  \xrightarrowd \mathcal{N}(\mathbf{0}, \boldsymbol{\Sigma}_g),
\]
where $\boldsymbol{\Sigma}_g$ has $(s,r)$-entry equal to the
asymptotic covariance $\lim_{N\to\infty} N \cdot
\Cov(\hat\gamma_{g,s},\hat\gamma_{g,r})$,
consistently estimated from sample second moments.

The OLS estimator of the polynomial coefficients is:
$\hat{\mathbf{c}}_g
= (\mathbf{V}^\top\mathbf{V})^{-1}\mathbf{V}^\top
\hat{\boldsymbol{\gamma}}_g^{\mathrm{pre}}$,
which is a continuous (in fact linear) function of
$\hat{\boldsymbol{\gamma}}_g^{\mathrm{pre}}$.
By A4, $\mathbf{V}^\top\mathbf{V}$ is invertible.
By the delta method (linear functions preserve normality):
\[
  \sqrt{N}(\hat{\mathbf{c}}_g - \mathbf{c}_g)
  \xrightarrowd
  \mathcal{N}(\mathbf{0}, \boldsymbol{\Sigma}_c),
  \quad
  \boldsymbol{\Sigma}_c
  = (\mathbf{V}^\top\mathbf{V})^{-1}\mathbf{V}^\top
  \boldsymbol{\Sigma}_g
  \mathbf{V}(\mathbf{V}^\top\mathbf{V})^{-1}.
\]
Since
$\hat\gamma_{g,t}(0)
= \mathbf{v}_t^\top\hat{\mathbf{c}}_g$,
\[
  \sqrt{N}(\hat\gamma_{g,t}(0) - \gamma_{g,t}(0))
  \xrightarrowd
  \mathcal{N}(0, \mathbf{v}_t^\top\boldsymbol{\Sigma}_c\mathbf{v}_t).
\]

\noindent\emph{Step 3: Joint CLT for pre- and post-treatment gaps
and corrected variance.}

Both $\hat\gamma_{g,t}$ (for $t \geq g$) and
$\hat\gamma_{g,t}(0) = \mathbf{v}_t^\top\hat{\mathbf{c}}_g$
(which depends on pre-treatment gaps through
$\hat{\boldsymbol{\gamma}}_g^{\mathrm{pre}}$)
involve outcomes of the same never-treated units observed at
different time periods.
Under A1, units are independent across each other, but the
same unit contributes to $\bar Y_{\infty,t}$ for all $t$,
so $\hat\gamma_{g,t}$ and $\hat\gamma_{g,t}(0)$ are
\emph{not} asymptotically independent.

Consider the full stacked vector of pre- and post-treatment gaps
for cohort $g$:
\[
  \hat{\boldsymbol{\gamma}}_g^{\mathrm{full}}
  = \begin{pmatrix}
      \hat{\boldsymbol{\gamma}}_g^{\mathrm{pre}} \\
      \hat{\boldsymbol{\gamma}}_g^{\mathrm{post}}
    \end{pmatrix}
  \in \mathbb{R}^{m_g + (T - g + 1)}.
\]
By the multivariate CLT applied jointly under A1--A2$'$:
\[
  \sqrt{N}\bigl(
    \hat{\boldsymbol{\gamma}}_g^{\mathrm{full}}
    - \boldsymbol{\gamma}_g^{\mathrm{full}}
  \bigr)
  \xrightarrowd
  \mathcal{N}\!\left(\mathbf{0},\;
  \boldsymbol{\Sigma}_g^{\mathrm{full}}\right),
  \quad
  \boldsymbol{\Sigma}_g^{\mathrm{full}} =
  \begin{pmatrix}
    \boldsymbol{\Sigma}_{qq} & \boldsymbol{\Sigma}_{qp} \\
    \boldsymbol{\Sigma}_{pq} & \boldsymbol{\Sigma}_{pp}
  \end{pmatrix},
\]
where $\boldsymbol{\Sigma}_{pp}$ is the covariance matrix of
post-treatment gaps, $\boldsymbol{\Sigma}_{qq}$ of pre-treatment
gaps, and the cross-block
\[
  [\boldsymbol{\Sigma}_{pq}]_{t,s}
  = \frac{\Cov(Y_{it}(\infty), Y_{is}(\infty)\mid G_i=g)}{\pi_g}
  + \frac{\Cov(Y_{it}(\infty), Y_{is}(\infty)\mid G_i=\infty)}
         {\pi_\infty}
\]
for $t \geq g$, $s < g$, is generically non-zero (A2$'$
guarantees finiteness).
The hat vector for the polynomial projection is
$\mathbf{h}_t = \mathbf{V}(\mathbf{V}^\top\mathbf{V})^{-1}
\mathbf{v}_t \in \mathbb{R}^{m_g}$.
Writing
$\widehat{\mathrm{ATT}}^{(p_g)}(g,t)
= \hat\gamma_{g,t} - \mathbf{v}_t^\top\hat{\mathbf{c}}_g
= \hat\gamma_{g,t} - \mathbf{h}_t^\top
\hat{\boldsymbol{\gamma}}_g^{\mathrm{pre}}$,
the delta method applied to the joint CLT gives:
\[
  \sqrt{N}\bigl(
    \widehat{\mathrm{ATT}}^{(p_g)}(g,t) - \ATTgt
  \bigr)
  \xrightarrowd
  \mathcal{N}(0, V_{g,t}),
\]
where the correct variance is:
\[
  V_{g,t}
  = \sigma^2_{g,t}
  + \mathbf{v}_t^\top\boldsymbol{\Sigma}_c\mathbf{v}_t
  - 2\,\mathbf{v}_t^\top(\mathbf{V}^\top\mathbf{V})^{-1}
    \mathbf{V}^\top\boldsymbol{\sigma}_{g,t},
\]
with $\boldsymbol{\sigma}_{g,t}$ the vector of asymptotic
covariances $(\lim_{N\to\infty} N \cdot
\Cov(\hat\gamma_{g,s}, \hat\gamma_{g,t}))_{s < g}$
between each pre-treatment gap and the post-treatment gap at $t$.
The third term corrects for the shared never-treated control group
and was absent from earlier drafts; its sign depends on the
temporal correlation structure of the never-treated outcomes and
cannot be determined a priori.

\smallskip
\noindent\textbf{Part (ii): Aggregate ATT.}

The aggregate is
$\hat\theta = \sum_{g,t\geq g}
w_{g,t}\,\widehat{\mathrm{ATT}}^{(p_g)}(g,t)$.
Because all cohorts share the same never-treated control group,
sampling variation in the common control enters every cohort's
ATT estimator.
The covariance matrix of the joint vector
$(\widehat{\mathrm{ATT}}^{(p_g)}(g,t))_{g,t\geq g}$
is therefore \emph{not} block-diagonal by cohort:
cross-cohort covariance terms arise whenever
$g \neq g'$ through the shared control.
The correct aggregate variance is:
\[
  V = \sum_{g,t\geq g}\sum_{g',t'\geq g'}
  w_{g,t}\,w_{g',t'}\,
  \Cov\!\bigl(\widehat{\mathrm{ATT}}^{(p_g)}(g,t),\,
  \widehat{\mathrm{ATT}}^{(p_{g'})}(g',t')\bigr),
\]
including all cross-cohort terms.
Since $\hat\theta$ is a linear function of the jointly
asymptotically normal vector of all gap estimates, asymptotic
normality of $\hat\theta$ follows from the delta method.
The full variance $V$ is consistently estimated by the cluster
bootstrap (Part iii).

\smallskip
\noindent\textbf{Part (iii): Bootstrap consistency.}

The cluster bootstrap resamples entire units with replacement,
including both treated and never-treated units simultaneously.
Because the same never-treated units enter the estimation for all
cohorts and at all time periods, resampling at the unit level
correctly replicates the joint sampling variation of the shared
control group across cohorts and across pre- and post-treatment
periods.
This captures all cross-cohort and pre/post covariance terms
identified in Steps 2--3 and Part (ii) above, without requiring
the researcher to compute them analytically.
Consistency of the cluster bootstrap for the aggregate variance $V$
follows by the bootstrap CLT for sample means under A1--A3 and
A2$'$; see \citet{callaway2021difference} Theorem~3.6
for the analogous argument.
\hfill$\square$

\section{Extrapolation Horizon Simulation}
\label{app:horizon}

This appendix quantifies how DD[$p$] performance changes with the post-treatment extrapolation horizon.

\textbf{Design.}
The panel has $N=300$ units, $T=25$ periods, two cohorts
($g=6,8$) with 5 and 7 pre-treatment periods respectively,
and 40\,\% never-treated.
True ATT $= 0.5$ throughout.
Two DGPs are considered.
\emph{DGP-Linear}: untreated outcomes follow
$Y_{it}(\infty) = \alpha_i + \beta_g t + \lambda_t + \varepsilon_{it}$
with $\beta_g \sim \mathcal{N}(0, 0.40)$ --- true Parallel[2].
\emph{DGP-Quadratic}: $Y_{it}(\infty) = \alpha_i + \beta_g t + c_g t^2 + \lambda_t + \varepsilon_{it}$ with $c_g \sim \mathcal{N}(0, 0.001)$ --- slight curvature that is indistinguishable from a linear trend in the short pre-treatment window but accumulates extrapolation error as $\tau$ grows.
Results are based on 500 replications.

\textbf{Results.}
Table~\ref{tab:horizon} reports bias and RMSE for DD[1], DD[2],
and DD[3] at selected post-treatment horizons under the quadratic
DGP.
Figure~\ref{fig:horizon} displays the full trajectory.

\begin{table}[H]
\centering
\caption{Extrapolation Horizon: Bias and RMSE under Slight Quadratic Curvature}
\label{tab:horizon}
\small
\begin{tabular}{lcccccc}
\toprule
 & \multicolumn{2}{c}{\textbf{DD[1]}}
 & \multicolumn{2}{c}{\textbf{DD[2]}}
 & \multicolumn{2}{c}{\textbf{DD[3]}} \\
\cmidrule(lr){2-3}\cmidrule(lr){4-5}\cmidrule(lr){6-7}
$\tau$ & Bias & RMSE & Bias & RMSE & Bias & RMSE \\
\midrule
1  & $-$0.006 & 0.154 & $-$0.002 & 0.142 & $-$0.001 & 0.224 \\
3  & $-$0.003 & 0.208 &    0.005 & 0.190 &    0.006 & 0.571 \\
5  & $-$0.015 & 0.258 & $-$0.003 & 0.234 & $-$0.003 & 1.114 \\
10 & $-$0.018 & 0.408 &    0.002 & 0.366 & $-$0.003 & 3.298 \\
15 & $-$0.031 & 0.565 & $-$0.002 & 0.506 & $-$0.016 & 6.600 \\
\bottomrule
\end{tabular}
\begin{minipage}{0.94\textwidth}\footnotesize\medskip
\textit{Notes}: 500 replications, $N=300$, $T=25$, true ATT $=0.5$.
DGP: slight quadratic curvature $c_g \sim \mathcal{N}(0, 0.001)$,
indistinguishable from a linear pre-trend with 5--7 pre-treatment
observations. DD[2] bias remains small at all horizons; RMSE grows
from 0.14 at $\tau=1$ to 0.51 at $\tau=15$. DD[3] RMSE grows
much faster under this DGP, reflecting the cost of fitting a
quadratic to only 5--7 pre-treatment observations; with longer
pre-treatment series this cost would be expected to diminish.
\end{minipage}
\end{table}

\begin{figure}[H]
  \centering
  \includegraphics[width=0.96\textwidth]{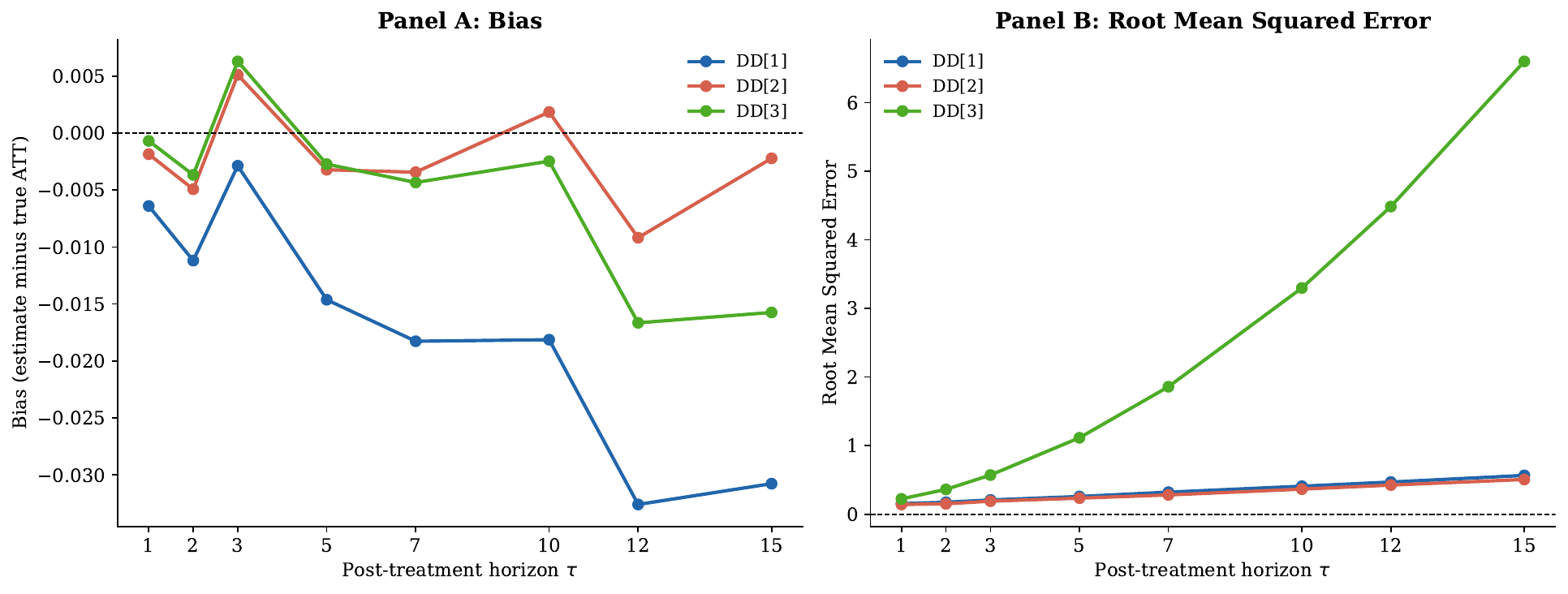}
  \caption{Extrapolation Bias and RMSE by Post-Treatment Horizon}
  \label{fig:horizon}
  \begin{minipage}{0.94\textwidth}\footnotesize\medskip
    \textit{Notes}: Panel A shows bias; Panel B shows RMSE.
    DGP has slight quadratic curvature ($c_g \sim \mathcal{N}(0,0.001)$)
    indistinguishable from a linear pre-trend.
    500 replications, $N=300$, $T=25$, true ATT $=0.5$.
    Bold borders identify the estimator matching the true DGP at
    short horizons.
  \end{minipage}
\end{figure}

\textbf{Interpretation.}
DD[2] bias remains small (maximum 0.009) across all horizons even
under slight curvature, because the linear fit approximately absorbs
the curvature in the pre-period.
The main cost is variance: RMSE grows from 0.14 at $\tau=1$
to 0.51 at $\tau=15$.
DD[2] is reliable for short post-treatment windows but variance grows with horizon.
DD[3] shows the most pronounced degradation under this short-pre-period
DGP; researchers with longer pre-treatment series should consult the
pre-period $R^2$ and sensitivity table before ruling out higher orders.

\section{Stata Replication Code and Applied Workflow}
\label{app:stata}

Complete replication code and companion Stata command implementing the DD[$p$] estimator is available at \url{https://github.com/zecharias2019/anddp}.

\begin{remark}[Applied workflow checklist]
\label{rem:workflow}
For practitioners applying DD[$p$] to a new dataset, the
following sequence is recommended.
\begin{enumerate}[label=Step \arabic*.]
  \item \emph{Setup.}
    Ensure \texttt{gvar} equals the first treatment year for
    treated units and 0 for never-treated.
    Unit and time fixed effects are absorbed implicitly through
    the gap construction; there is no need to add them
    explicitly.

  \item \emph{Establish the baseline.}
    Run \texttt{csdid} and \texttt{estat pretrend}.
    If the joint pre-trend test does not reject, standard
    Parallel[1] is adequate and \texttt{csdid} results can
    be reported as the primary specification.

  \item \emph{Select the polynomial order.}
    If Parallel[1] is rejected, or as a matter of course,
    run \texttt{anddp} with \texttt{maxorder()} to let the
    data guide order selection:
    \begin{quote}
      \texttt{anddp y, ivar() time() gvar() maxorder(4) reps(999)}
    \end{quote}
    The command tests each cohort's pre-treatment gap
    independently and selects the lowest order not rejected
    by the sequential F-test, reporting cohort-specific
    diagnostic information (pre-period $R^2$, F-test p-value,
    effective order, and weight).
    Stability of estimates across the order sensitivity table
    and pre-period $R^2$ values are the primary practical
    guides; a high $R^2$ and stable aggregate ATT across
    orders support the selected specification.

    Alternatively, researchers may impose a specific order
    when the pre-treatment data clearly support it:
    \begin{quote}
      \texttt{anddp y, ivar() time() gvar() order(2) reps(999)}
    \end{quote}
    The choice \texttt{order(p)} fixes order $p$ for all
    cohorts (subject to the feasibility cap $p_g \leq m_g-1$)
    and is appropriate when theory or the diagnostic evidence
    motivates a particular specification.
    Beyond the formal diagnostics, plotting the pre-treatment
    gap directly is often informative: a steadily rising or
    falling trajectory supports order 2, while visible curvature
    suggests order 3.

  \item \emph{(Optional) Covariate adjustment.}
    If treated and never-treated units differ on observable
    time-varying characteristics, add covariates:
    \begin{quote}
      \texttt{anddp y, ivar() time() gvar() order(2)
      covariates(x1 x2) reps(999)}
    \end{quote}
    Covariates enter the polynomial regression alongside time;
    the counterfactual extrapolates both the polynomial and
    the observed post-treatment covariate gaps.
    The identifying assumption becomes Parallel[$p$]
    conditional on the covariate gaps.

  \item \emph{Sensitivity across orders.}
    Report the aggregate ATT under orders 1, 2, and 3
    side by side using the sensitivity table that
    \texttt{anddp} produces automatically.
    Different orders may perform differently depending on
    the dataset: when pre-treatment series are long and
    $R^2$ supports it, a higher order can be the most
    credible primary specification.
    Note which cohorts have fewer than three pre-treatment
    periods, as inference for those cells is fragile.

  \item \emph{Check cohort sizes for reliable inference.}
    Inspect the Weight column in the diagnostics table, which
    reflects each cohort's share of treated units.
    When a cohort contains very few units --- a single state,
    in the extreme, as with the 2017 Medicaid cohort below ---
    standard cluster bootstrap inference for that cohort's ATT
    is unreliable: the bootstrap distribution becomes close to
    discrete, and the reported SE understates genuine
    uncertainty.  This is a general limitation of cluster
    bootstrap inference with few clusters, not specific to
    DD[$p$].  The aggregate ATT is typically more robust, since
    small cohorts carry low population weight, but cohort-level
    inference for small cohorts should be interpreted with
    caution.  The \texttt{wcbootstrap} option in \texttt{anddp}
    addresses this directly: it perturbs cluster-level residuals
    with random signs rather than resampling clusters, and
    remains valid for inference even when a cohort has only one
    or two units.

  \item \emph{Long post-treatment windows.}
    If the post-treatment window exceeds ten periods,
    inspect Appendix~\ref{app:horizon} for evidence on
    how extrapolation reliability varies with horizon.
    Complementing the point estimates with sensitivity bounds
    \citep{rambachan2023more} is advisable when the horizon
    is long relative to the pre-treatment series.
\end{enumerate}
\end{remark}

\section{Comparison of Parallel Trends Relaxations}
\label{app:comparison}

Table~\ref{tab:comparison} summarizes the relationship
between standard DiD, the \citet{dobkin2018economic}
linear-trend correction, and the DD[$p$] estimator
proposed in this paper.

\begin{table}[H]
\centering
\caption{Comparison of Parallel Trends Relaxations}
\label{tab:comparison}
\small
\setlength{\tabcolsep}{5pt}
\renewcommand{\arraystretch}{1.4}
\begin{tabular}{p{2.8cm}p{3.2cm}p{3.2cm}p{3.2cm}}
\toprule
 & \textbf{Parallel[1]} (standard DiD)
 & \textbf{Dobkin et al.\ (2018)}
 & \textbf{DD[$p$] (this paper)} \\
\midrule
\textit{Assumption} &
  Untreated gap is flat: $\Delta Y_{g,t}(\infty)=0$ &
  Untreated gap is linear: $\Delta^2 Y_{g,t}(\infty)=0$
  (Parallel[2]) &
  $p$-th difference of untreated gap is equal across
  groups: $\Delta^p Y_{g,t}(\infty) = 0$ \\
\textit{Slope estimated from} &
  No slope &
  Pre-treatment data only (post-treatment dummies
  are saturated) &
  Pre-treatment data only \\
\textit{Order selection} &
  Fixed at $p=1$ &
  Fixed at $p=2$; chosen by inspection &
  Sequential test selects smallest $p$ not rejected
  by pre-treatment data \\
\textit{Staggered adoption} &
  Yes (Callaway--Sant'Anna) &
  No; single treatment timing &
  Yes; formal aggregation across cohorts with
  different feasible orders \\
\textit{Aggregation} &
  CS weights &
  Not addressed &
  Theorem~\ref{thm:aggregation}; valid under
  cohort-heterogeneous orders \\
\textit{Extrapolation risk} &
  Flat gap assumed to continue &
  Linear trend extrapolated; horizon not addressed &
  Polynomial extrapolated; Remark~\ref{rem:roth_poly}
  quantifies risk by horizon \\
\bottomrule
\end{tabular}
\begin{minipage}{0.98\textwidth}\footnotesize\medskip
\textit{Notes}: Dobkin et al.'s linear-trend correction is
formally equivalent to the $p=2$ case of Parallel[$p$] applied to a single cohort (Jonathan Roth, personal communication, 2026).
All three approaches rely on extrapolation into the
post-treatment period; none is testable post-treatment.
\end{minipage}
\end{table}

\end{document}